\def\A{{\cal A}}
	\def\eps{{\varepsilon}}	

\newif\ifsocg
\socgfalse
\ifsocg
\documentclass[a4paper,UKenglish]{socg-lipics-v2019}
\else
\documentclass[11pt]{article} 
\usepackage{amsfonts} 
\usepackage{graphicx} 
\usepackage{amsthm} 
\usepackage{amsmath} 
\usepackage{hyperref} 
\fi

\usepackage{romannum}

\usepackage{microtype}

\title{On the Complexity of the \texorpdfstring{$k$}{k}-Level in Arrangements of Pseudoplanes \thanks{Work on this paper was supported by Grants 892/13 and 260/18 from the Israel Science Foundation. Work on this paper by Micha Sharir was also supported by Grant G-1367-407.6/2016 from the German-Israeli Foundation for Scientific Research and Development, by the Blavatnik Research Fund in Computer Science at Tel Aviv University, and by the Israeli Centers of Research Excellence (I-CORE) program (Center No.~4/11).}}

\ifsocg
\author{Micha Sharir \thanks{ddd}}{School of Computer Science, Tel Aviv University, Tel Aviv, Israel}{michas@post.tau.ac.il}{http://www.cs.tau.ac.il/~michas/}{}

\author{Chen Ziv}{School of Computer Science, Tel Aviv University, Tel Aviv, Israel}{henziv@gmail.com}{}{}
 
\authorrunning{M. Sharir and C. Ziv}

\Copyright{Micha Sharir and Chen Ziv}



\category{}

\funding{}

\EventEditors{Gill Barequet and Yusu Wang}
\EventNoEds{2}
\EventLongTitle{35th International Symposium on Computational Geometry (SoCG 2019)}
\EventShortTitle{SoCG 2019}
\EventAcronym{SoCG}
\EventYear{2019}
\EventDate{June 18--21, 2019}
\EventLocation{Portland, United~States}
\EventLogo{socg-logo}
\SeriesVolume{129}
\ArticleNo{48} 

\else
\author{
	Micha Sharir\\
	School of Computer Science, Tel Aviv University, Tel Aviv, Israel\\
	\href{mailto:michas@post.tau.ac.il}{\texttt{michas@post.tau.ac.il}}
	\and
Chen Ziv\\
School of Computer Science, Tel Aviv University, Tel Aviv, Israel\\
\href{mailto:henziv@gmail.com}{\texttt{henziv@gmail.com}}
}

\fi
\begin{document}

\newtheorem{theorem}{Theorem}[section]
\newtheorem{lemma}[theorem]{Lemma}
\newtheorem{definition}[theorem]{Definition} 

\pagenumbering{arabic}

\newcommand{\rnum}[1]{\uppercase\expandafter{\romannumeral #1\relax}}

\maketitle

\begin{abstract}
A classical open problem in combinatorial geometry is to obtain tight asymptotic bounds on the maximum number of \emph{$k$-level} vertices in an arrangement of $n$ hyperplanes in $\mathbb{R}^d$ (vertices with exactly $k$ of the hyperplanes passing below them). This is essentially a dual version of the \emph{$k$-set problem}, which, in a primal setting, seeks bounds for the maximum
number of \emph{$k$-sets} determined by $n$ points in $\mathbb{R}^d$, where a $k$-set is a subset of size $k$ that can be separated from its complement by a hyperplane. The $k$-set problem is still wide open even in the plane. In three dimensions, the best known upper and lower bounds are, respectively, $O(nk^{3/2})$~\cite{SST99} and $nk\cdot 2^{\Omega(\sqrt{\log k})}$~\cite{T00}.

In its dual version, the problem can be generalized by replacing hyperplanes by other families of surfaces (or curves in the planes). Reasonably sharp bounds have been obtained for curves in the plane~\cite{SZ16, TT03}, but the known upper bounds are rather weak for more general surfaces, already in three dimensions, except for the case of triangles~\cite{AACS98}.
The best known general bound, due to Chan~\cite{C12} is $O(n^{2.997})$, for families of surfaces that satisfy certain (fairly weak) properties.

In this paper we consider the case of \emph{pseudoplanes} in $\mathbb{R}^3$ (defined in detail in the introduction), and establish the upper bound $O(nk^{5/3})$ for the number of $k$-level vertices in an arrangement of $n$ pseudoplanes. The bound is obtained by establishing suitable (and nontrivial) extensions of dual versions of classical tools that have been used in studying the primal $k$-set problem, such as the Lov\'asz Lemma and the Crossing Lemma.

 \end{abstract}

\section{Introduction}
\label{sec: inroduction}

Let $\Lambda$ be a set of $n$ non-vertical planes (resp., pseudoplanes, as will be formally defined shortly) in $\mathbb{R}^3$, in general
position. We say that a point $p$ lies at \textit{level} $k$ of the arrangement $\A(\Lambda)$, and write $\lambda(p)=k$, if exactly $k$ planes
(resp., pseudoplanes) of $\Lambda$ pass below $p$. The \textit{$k$-level} of $\A(\Lambda)$ is the closure of the set of points that lie on the
surfaces of $\Lambda$ and are at level $k$. Our goal is to obtain an upper bound on the complexity of the $k$-level of $\A(\Lambda)$, which is
measured by the number of vertices of $\A(\Lambda)$ that lie at level $k$. (The level may also contain vertices at level $k-1$ or $k-2$, but we
ignore this issue---it does not affect the worst-case asymptotic bound that we are after.) Using a standard duality transform that preserves the above/below relationship (see, e.g., \cite{Ed87}), the case of planes is the dual version of the following variant of the \textit{$k$-set} problem: given a set of $n$ points in $\mathbb{R}^3$ in general position, how many triangles spanned by $P$ are such that the plane supporting the triangle has exactly $k$ points of $P$ below it? We refer to these triangles as \textit{k-triangles}. This has been studied by Dey and Edelsbrunner~\cite{DE94}, in 1994, for the case of \textit{halving triangles}, namely $k$-triangles with $k=(n-3)/2$ (and $n$ odd). They have shown that the number of halving triangles is $O(n^{8/3})$. In 1998, Agarwal et al.~\cite{AACS98} generalized this result for $k$-triangles, for arbitrary $k$, showing that their number is $O(nk^{5/3})$, using a probabilistic argument. In 1999, Sharir, Smorodinsky and Tardos~\cite{SST99} improved the upper bound for the number of $k$-triangles in $S$ to $O(nk^{3/2})$. 
	Chan~\cite{C10} has adapted a dualized view of the technique of \cite{SST99} in order to study the bichromatic $k$-set problem: given two sets $R$ and $B$ of points in $\mathbb{R}^2$ of total size $n$ and an integer $k$, how many subsets of the form $(R \cap h) \cup (B \setminus h)$ can have size exactly $k$, over all halfplanes $h$?
This problem arises when we estimate the number of vertices at level $k$, in an arrangement of $n$ planes in 3-space, that lie on one specific plane.

	The three-dimensional case extends the more extensively studied planar case. In its primal setting, we have a set $S$ of $n$ points in
the plane in general position, and a parameter $k<n$, and we seek bounds on the maximum number of \emph{$k$-edges}, which are segments spanned
by pairs of points of $S$ so that one of the halfplanes bounded by the line supporting the segment, say the lower halfplane, contains exactly
$k$ points of $S$. In the dual version, we seek bounds on the maximum number of vertices of an arrangement of $n$ nonvertical lines in general
position that lie at level $k$. The best known upper bound for this quantity, due to Dey~\cite{D98}, is $O(nk^{1/3})$, and the best known lower
bound, due to T\'{o}th~\cite{T00} is $ne^{\Omega( \sqrt{\ln k} )}$.
(Nivasch~\cite{N00} has slightly improved this bound, to $\Omega(ne^{\sqrt{\ln4} \cdot \sqrt{\ln n}} / \sqrt{\ln n})$, for the case of halving edges.)

	In this paper we consider the dual version of the problem in three dimensions, where the points are mapped to planes, and the $k$-triangles are mapped to vertices of the arrangement of these planes at level $k$. We translate parts of the machinery developed in~\cite{SST99} to the dual setting, and then extend it to handle the case of pseudoplanes. Since there is a lot of similarity between some of the main techniques and ideas of the case of planes and the case of pseudoplanes, we omit some of the details from the case of planes and mainly focus on, and present the full detailed version for, the case of pseudoplanes.
	
	 In the primal setting (for the case of planes), we have a set $S$ of $n$ points in $\mathbb{R}^3$ in general position, and the set $T$ of $k$-triangles spanned by $S$. We say that triangle $\Delta_1$ \textit{crosses} another triangle $\Delta_2$ if the triangles share exactly one vertex, and the edge opposite to that vertex in $\Delta_1$ intersects the interior of $\Delta_2$. Denote the number of ordered pairs of crossing $k$-triangles by $X^k$. The general technique in~\cite{SST99} is to establish an upper bound and a lower bound on $X^k$, and to combine these two bounds to derive an upper bound for the number of $k$-triangles in $S$.
	
	The upper bound in \cite{SST99} is based on the 3-dimensional version of the Lov\'asz Lemma, as in \cite{BFL88}: Any line crosses at most $O(n^2)$ interiors of $k$-triangles. The lemma follows from the main property of the set $T$, which is its \textit{antipodality}. Informally, the property asserts that for each pair of points $a,b \in S$, the $k$-triangles having $ab$ as an edge form an antipodal system, in the sense that for any pair $\Delta_{abc}, \Delta_{abd}$ of such triangles that are consecutive in the circular order around $ab$, the dihedral wedge that is formed by the two halfplanes that contain $\Delta_{abc}, \Delta_{abd}$, and are bounded by the line through $ab$, has the property that its \textit{antipodal wedge}, formed by the two complementary halfplanes within the planes supporting $\Delta_{abc}, \Delta_{abd}$, contains a point $e \in S$ such that $\Delta_{abe}$ is also a $k$-triangle; See Figure \ref{fig: antipodality}.

	\begin{figure}[!htbp]
		\centering
  		\includegraphics[width=2.5in]{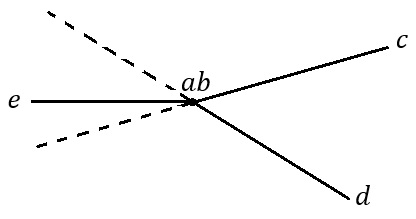}
  		\caption{\small\sf The antipodality property of $k$-triangles: the edge $ab$ is drawn from a side-view as a point. The triangles $\Delta_{abc}, \Delta_{abd}$ and $\Delta_{abe}$, drawn as segments in the figure, are all $k$-triangles, where $e$ is contained in the antipodal wedge formed by the two complementary halfplanes within the planes supporting $\Delta_{abc}, \Delta_{abd}$.}  	
  		\label{fig: antipodality}
	\end{figure}

	To obtain a lower bound on $X^k$, the technique in	\cite{SST99} defines, for each $a \in S$, a graph $G_a=(V_a,E_a)$ drawn in a horizontal plane $h_a^+$ slightly above $a$, whose edges are, roughly, the cross-sections of the $k$-triangles incident to $a$ with the plane. See Figure \ref{fig: classic_lower_bound_graph} for an illustration. The analysis in \cite{SST99} shows that $G_a$ inherits the antipodality property of the $k$-triangles, and uses this fact to decompose $G_a$ into a collection of convex chains, and to estimate the number of crossings between the chains. Summing these bounds over all $a \in S$, the lower bound on $X^k$ follows.

	\begin{figure}[!htbp]
		\centering
  		\includegraphics[width=4in]{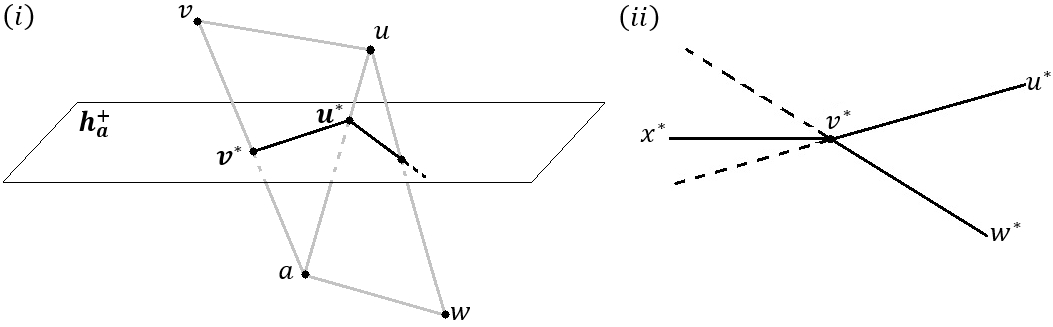}
  		\caption{\small\sf (\romannum{1}) The graph $G_a$ on $h_a^+$. $\Delta_{auw}, \Delta_{auv}$ are halving triangles. $uv$ is mapped to the segment $u^*v^*$, and $uw$ is mapped to a ray emanating from $u^*$ on $h^+_a$. (\romannum{2}) The antipodality property in $G_a$.}  	
  		\label{fig: classic_lower_bound_graph}
	\end{figure}

	We omit further details of the way in which these lower bounds are derived in \cite{SST99}, because, in the generalized dual version to the case of pseudoplanes that we present here, we use a weaker lower bound, which is based on a generalized dual version of the \textit{Crossing Lemma} (see~\cite{ACNS82}), and which is easier to extend to that case. Let $G=(V, E)$ be a simple graph, and define the \textit{crossing number} of $G$ as the minimum number of intersecting pairs of edges in any drawing of $G$ in the plane. In the primal setting, the Crossing Lemma asserts that any simple graph $G=(V,E)$ drawn in the plane, with $|E|>4|V|$, has crossing number at least\footnote{The constant of proportionality has been improved in subsequent works, but we will stick to this bound.} $\frac{|E|^3}{64|V|^2}$. Using this technique for deriving a lower bound on $X^k$, instead of the refined technique in \cite{SST99}, one can show that the number of $k$-triangles is $O(n^{8/3})$, or, with the additional technique of~\cite{AACS98}, $O(nk^{5/3})$.
	
	We now present the dual setting for the problem, where the input is a set $\Lambda$ of $n$ non-vertical planes in $\mathbb{R}^3$ in general position.
	
	\begin{definition}
	\label{def: corridor}
	Let $a,b,c\in\Lambda$. The open region between the lower envelope and the upper envelope of ${a, b, c}$ is called \emph{\textbf{the corridor of $a$, $b$, $c$}} and is denoted by \emph{\boldmath{$C_{a,b,c}$}}.
	\end{definition}
	
	The planes $a,b,c \in \Lambda$ divide the space into eight disjoint octant-like portions, and $C_{a,b,c}$ is the union of six of those portions, excluding the upper and the lower octants. We will be mostly interested in corridors $C_{a,b,c}$ for which the point $p_{a,b,c} = a \cap b \cap c$ (the unique vertex of $C_{a,b,c}$) is at level $k$. We will refer to such corridors as \emph{$k$-corridors}, and define $C^k$ as the collection of $k$-corridors in $\A(\Lambda)$; $k$-corridors serve as a dual version of $k$-triangles.

	\begin{definition}
	\label{def: immersed}
	We say that a corridor $C_1$ \emph{\textbf{is immersed}} in a corridor $C_2$ in $\A(\Lambda)$ if they share exactly one plane, and the intersection line of the other two planes of $C_1$ is fully contained in $C_2$. Let $X^k$ denote the number of ordered pairs of immersed corridors in $C^k$. 
	\end{definition}
	
	Immersion of $k$-corridors is the dual notion of crossings of $k$-triangles. Note that if a corridor $C_1$ is immersed in a corridor $C_2$, it cannot be that $C_2$ is also immersed in $C_1$ (this is a consequence of our general position assumption).	
	
	In Section \ref{sec: planes} we provide more details of this dual setup. We present the derivation of the upper bound on $X^k$, the
number of ordered pairs of immersed $k$-corridors, using a dual version of the Lov\'asz Lemma. The reason for doing this is twofold:
(i) this translation to the dual context, although routine in principle, is rather involved and nontrivial, and requires careful handling of
quite a few details, and (ii) it provides several basic technical ingredients that we will need to extend to the case of pseudoplanes.
	
	
	In Section \ref{sec: pseudoplanes} we consider in detail the extension to the case of pseudoplanes, which is our main topic of interest. In our context, a family $\Lambda$ of $n$ surfaces in $\mathbb{R}^3$ is a \textit{family of pseudoplane}, if, in addition to the generally acceptable definition of a pseudoplane family (namely, the surfaces are graphs of total bivariate continuous functions, and each triple of them intersect exactly once), it satisfies the following conditions:
	\renewcommand\labelenumi{(\roman{enumi})}
	\renewcommand\theenumi\labelenumi
	\begin{enumerate}
	\item The intersection of any pair of surfaces in $\Lambda$ is a connected $x$-monotone unbounded curve.
	\item The $xy$-projections of the set of all ${n\choose 2}$ intersection curves of the surfaces form a family of pseudolines in the
$xy$-plane.
\end{enumerate}
	
	 Both conditions, especially the second one, are nontrivial assumptions for a general pseudoplane family (although they trivially hold in the case of planes).
	 
	 We generalize the definition of $X^k$ for the number of ordered pairs of immersed $k$-corridors, where $\Lambda$ is a family of
pseudoplanes as above. We then generalize the analysis in Section \ref{sec: planes}, for obtaining an upper bound and a lower bound for $X^k$.
Comparing those bounds gives us the bound $O(n^{8/3})$ on the complexity of the $k$-level, which can then be refined into the
$k$-sensitive bound $O(nk^{5/3})$.

\section{The case of planes}
\label{sec: planes}

	In this section we present some ingredients of our technique for the simpler case of planes. We focus on the derivation of the upper bound on $X^k$ and the dual version of the Lov\'asz Lemma in $\mathbb{R}^3$. This might help readers to get familiar with the machinery of this paper, in the simpler, and easier to visualize, context of planes. On the face of it, in the case of planes this is just a translation to the dual setting of classical arguments used in the primal analysis of $k$-sets in three dimensions~\cite{DE94,SST99}. Still, this translation is fairly nontrivial, so getting familiar with it in this simple context might be helpful for absorbing the more general arguments in our analysis.
	
	We also discuss in this section the challenge of generalizing the technique to general pseudoplane arrangements, as defined above
(see also the beginning of Section~\ref{sec: pseudoplanes}).
For the planar case, the bound on the complexity of the $k$-level of an arrangement of lines
holds for pseudoline arrangements too, as shown by Tamaki and Tokuyama~\cite{TT03}. They have shown that Dey's combinatorial arguments (see \cite{D98}) for the case of arrangements of lines hold also for the case of pseudoline arrangements.
	The generalization in $\mathbb{R}^3$ from the case of planes to the case of pseudoplanes is more complicated, as the technique presented 
in Section~\ref{sec: pseudoplanes} 
indicates. In this section we provide an example where the technique used in the proof of the dual version of the Lov\'asz Lemma,
as presented here, does not easily apply for general arrangements of pseudoplanes.
	
	\subsection{A dual version of the Lov\'asz Lemma} 			
	\label{subsec: lovasz_lemma}
	
	Let $\Lambda$ be a set of $n$ planes in general position in $\mathbb{R}^3$. The following lemma is a dual variant of the antipodality of the set of $k$-triangles in the primal setup, as reviewed above.
	
	\begin{figure}[!htbp]
		\centering
  		\includegraphics[width=4in]{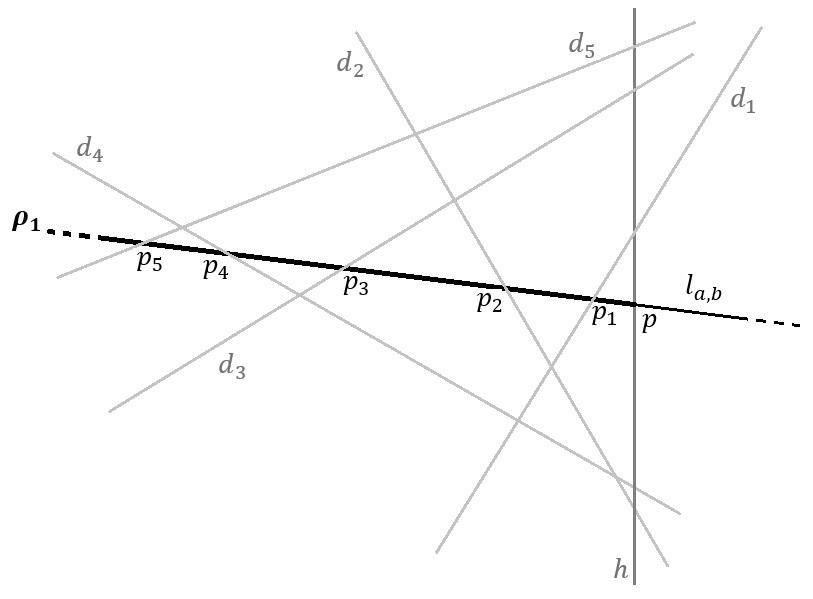}
  		\caption{\small\sf A $z$-vertical line $h$ that meets $l_{a,b} = a \cap b$ at a point $p$, and one of the rays $\rho_1$ that emanates from $p$ and is contained in $l_{a,b}$. In this case, $\{d_1, d_3, d_5\} \subseteq h_{up}$ and $\{d_2, d_4\} \subseteq h_{down}$, and $D_{2,4} = \{d_2,d_3,d_4\}$.} 
  		\label{fig: Hup_Hdown} 	
	\end{figure}
	
    \begin{lemma}
    \label{lemma: antipodality}
	Let $a, b\in\Lambda$, and let $l_{a,b} = a \cap b$ denote their intersection line. Let $h$ be a $z$-vertical line (orthogonal to the $xy$-plane) that intersects $l_{a,b}$ at some point $p$. Let $D = \Lambda \setminus \{a,b\}$, and $D^k = \{d \in D \mid C_{a,b,d} \in C^k\}$. Denote $h_{up} = \{d \in D \mid z_{d \cap h}>z_p\}$ and $h_{down} = \{d \in D \mid z_{d \cap h}<z_p\}$ (by choosing $p$ generically, we may assume that all these inequalities are indeed sharp).
	We then have $\Bigl| |h_{up} \cap D^k| - |h_{down} \cap D^k| \Bigr| \leq 2$.
	\end{lemma}	
	
	\begin{proof}	
	Denote one of the rays that emanates from $p$ and is contained in $l_{a,b}$ by $\rho_1$, and the other one by $\rho_2$, and denote $D_{\rho_1} = \big\{d\in D \mid d$ intersects $\rho_1 \big\}$, $D_{\rho_2} = \big\{d\in D \mid d$ intersects $\rho_2 \big\}$. Clearly, with a generic choice of $p$, $D_{\rho_1} \cup D_{\rho_2} = D$. Enumerate the planes in $D_{\rho_1}$ as $d_1,\ldots,d_{j}$, according to the order in which their respective intersection points with $\rho_1$, denoted $p_1,\ldots,p_{j}$, appear on $\rho_1$ in the direction from $p$ to infinity.
	Assume there are $1 \leq r<s \leq j$ such that $d_r, d_s \in D_{\rho_1} \cap ( h_{down} \cap D^k )$, and denote $D_{r,s} = \{d_i \in D_{\rho_1} \mid r \leq i \leq s\}$. Each $d_i \in h_{up} \cap D_{r, s}$ with $r<i<s$ is above $p_r$ and below $p_{s}$, and each $d_i \in h_{down} \cap D_{r, s}$ with $r<i<s$ is above $p_{s}$ and below $p_{r}$; the same also holds for $d_r$ and $d_s$, except that $d_r$ passes through $p_r$ and $d_s$ passes through $p_s$ (see Figure \ref{fig: Hup_Hdown}).
	
	It is easy to show that, for each point $p_i$, the following properties hold (see Figure \ref{fig: neighbors_level}):
	
\renewcommand\labelenumi{(\roman{enumi})}
\renewcommand\theenumi\labelenumi
\begin{enumerate}
  \item $\lambda(p_i)=\lambda(p_{i-1})-1$ if and only if both $d_{i-1}, d_i \in h_{down}$.
  \item $\lambda(p_i)=\lambda(p_{i-1})+1$ if and only if both $d_{i-1}, d_i \in h_{up}$.
  \item $\lambda(p_i)=\lambda(p_{i-1})$ if and only if one of $d_{i-1}, d_i$ is in $h_{up}$ and the other one is in $h_{down}$.
\end{enumerate}
	
	\begin{figure}[!htbp]
		\centering
  		\includegraphics[width=6in]{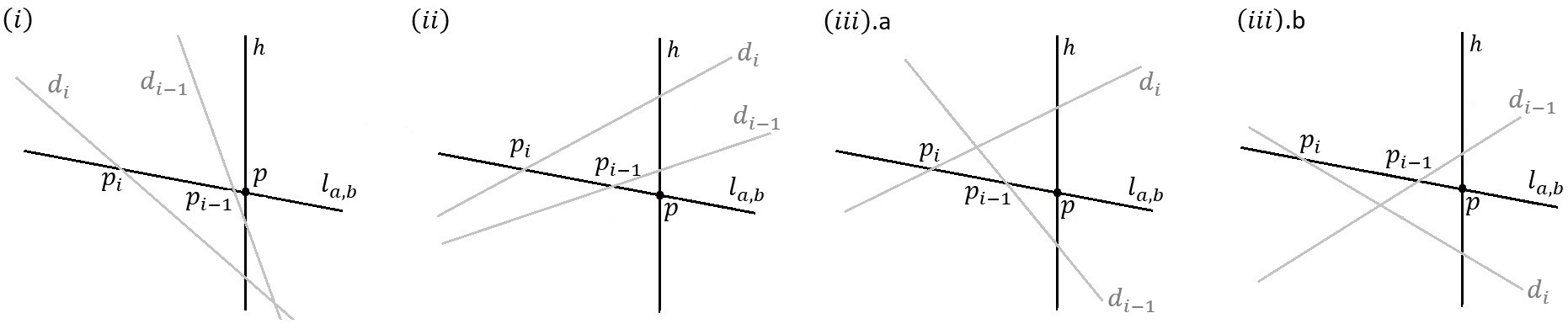}
  		\caption{\small\sf Both $d_i, d_{i-1}$ intersect the same ray of $l_{a,b}$ from $p$: (i) The case where both $d_i, d_{i-1} \in h_{down}$. (ii) The case where both $d_i, d_{i-1} \in h_{up}$. (iii).a-(iii).b The case where one of $d_i, d_{i-1}$ is in $h_{down}$, and the other one is in $h_{up}$.}  	
  		\label{fig: neighbors_level}
	\end{figure}
	
	We claim that there must exist a plane $d'\in D_{r,s}$ that is in $h_{up} \cap D^k$. If $d_{r+1} \in h_{up}$, by (iii) above, $\lambda(p_{r+1})=\lambda(p_{r})=k$. Since $d_{r+1} \notin h_{down}$, $r+1 <s$, and therefore $d_{r+1} \in h_{up} \cap D^k$.
	Otherwise, $d_{r+1} \in h_{down}$, and by (i) above, the level of $p_{r+1}$ is $k-1$. Note that $r+1<s$ because the level of $p_{s}$ is $k$. Because the level can change only by $0$, $+1$, or $-1$ between two consecutive points $p_i$, $p_{i+1}$, there must be a point $p_i \in D_{r, s} \subseteq D_{\rho_1}$, so that the level of $p_i$ is $k$ and the level of the previous point on $l_{a,b}$ is $k-1$, which means, by (ii) above, that $d_i \in h_{up}$. That is, between each pair $p_r, p_s \in \rho_1$ so that $d_r, d_s \in h_{down} \cap D^k$, there exists $p_i$ so that $d_i \in h_{up} \cap D^k$, and our claim is established (see Figure \ref{fig: up_between_downs}).
	
	Similarly, between each pair $p_{r}, p_{s} \in \rho_1$ so that $d_{r}, d_{s} \in h_{up} \cap D^k$, there exists $p_i$, for some $r<i<s$, so that $d_i \in h_{down} \cap D^k$. Both of these properties are easily seen to imply that \[ \Bigl| |h_{up} \cap D^k \cap D_{\rho_1}| - |h_{down} \cap D^k \cap D_{\rho_1}| \Bigr| \leq 1. \]
	The same reasoning applies to $\rho_2$, and yields $\Bigl| |h_{up} \cap D^k \cap D_{\rho_2}| - |h_{down} \cap D^k \cap D_{\rho_2}| \Bigr| \leq 1$. Thus, $\Bigl| |h_{up} \cap D^k| - |h_{down} \cap D^k| \Bigr| \leq 2$.
	\hfill
	\end{proof}

	\begin{figure}[!htbp]
		\centering
  		\includegraphics[width=6in]{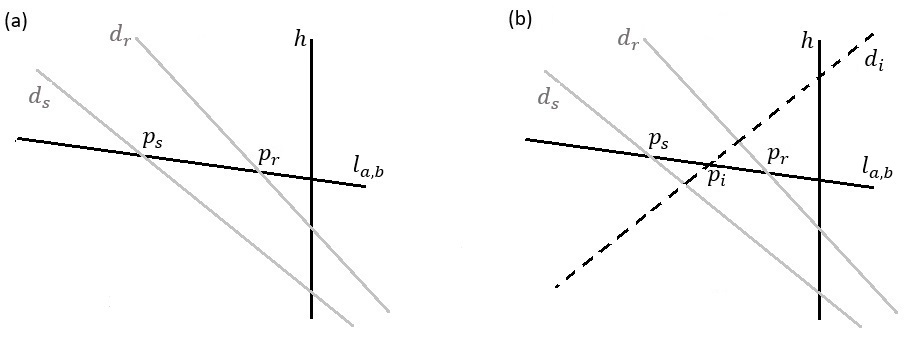}
  		\caption{\small\sf(a) The case where both $d_r, d_s \in h_{down} \cap D^k$. (b) There exists an in-between plane $d_i$ that belongs to $h_{up} \cap D^k$.}
  		\label{fig: up_between_downs}  	
	\end{figure}
	
	Note that if we were to place $h$ at infinity, the difference would have been only $\pm 1$ (as is the case in the primal setup).
	
	We now prove the promised dual version of the Lov\'asz Lemma.

	\begin{lemma}
	\label{lemma: lovasz}
	Any nonvertical line that is not parallel to any of the planes of $\Lambda$ and does not intersect any of the edges of $\A(\Lambda)$, is fully contained in at most $n(n-1)/2$ corridors in $C^k$.
	\end{lemma}

	\begin{proof}
	Let $l_0$ be a line as in the lemma. Consider the vertical plane through $l_0$, and let $l_1$ be a parallel line to $l_0$ contained in that plane, that is not contained in any of the corridors in $C^k$. We can find such a line, for example, by translating $l_0$ sufficiently far upwards (in the positive $z$-direction). For each $0 \leq \tau \leq 1$, define $l_{\tau}=\tau l_0 + (1-\tau) l_1$. As $\tau$ increases from $0$ to $1$, $l_{\tau}$ translates from $l_1$ down to $l_0$. During the translation, we maintain an upper bound on the number of corridors from $C^k$ that the translating line is fully contained in, until we reach $l_0$ at $\tau=1$, and argue that this bound remains $O(n^2)$.
	
	At the beginning of the process, the number of corridors from $C^k$ that $l_1$ is contained in is $0$. We say that the line $l_{\tau'}$ is \emph{about to enter} (resp., \emph{exit}) the corridor $C_{a,b,d}$, if there is $\eps'>0$ so that for every $0<\eps \leq \eps'$, $l_{\tau'+\eps}$ (resp., $l_{\tau'-\eps}$) is fully contained in the interior of $C_{a,b,d}$, but this does not hold for $l_{\tau'}$ itself.
	In order to reach a position where it is fully contained in a corridor $C_{a,b,d}$, during the translation downwards in the negative $z$-direction, the translating line $l_{\tau}$ has to reach a position where it is about to enter $C_{a,b,d}$.
	 If at some moment $l_{\tau'}$ reaches a position at which it is about to enter $C_{a,b,d}$, one of the following two situations must occur:
\begin{enumerate}
	\item One of the planes $a,b,d$ contains $l_{\tau'}$.
	\item There are two planes among $\{a,b,d\}$, say they are $a$ and $b$, such that $l_{\tau'}$ touches the line $a\cap b$ at some point $p$, at a position that lies on $\partial C_{a,b,d}$, $l_{\tau'}$ lies below $a$ on one side of $p$ and below $b$ on the other side, and $p$ lies above the third plane $d$.
\end{enumerate}

	Conversely, if the properties in (ii) hold, $l_{\tau'}$ is about to enter $C_{a,b,d}$.
	
	Since $l_0$ is not parallel to any plane of $\Lambda$, and $l_{\tau}$ is parallel to $l_0$, $l_{\tau}$ cannot be contained in any plane of $\Lambda$. Hence, the first scenario cannot happen. Assume that the second case occurs. Since, immediately below $l_{\tau'}$, the translating $l_{\tau}$ starts being fully contained in the interior of $C_{a,b,d}$, it follows that $l_{\tau'}$ must pass through an edge of $C_{a,b,d}$, which is contained in an intersection line of two of the planes $a,b,d$. Assume without loss of generality that $l_{\tau'}$ touches $a\cap b$, at some point $p$. Moreover, just before reaching this position, $l_{\tau}$ has a portion that lies above the upper envelope of $a$ and $b$, and this portion shrinks to the single point $p=l_{\tau'} \cap (a \cap b)$, for otherwise crossing $a \cap b$ would not make $l_{\tau}$ being fully contained in $C_{a,b,d}$. It follows that $l_{\tau'}$ is fully contained in the \lq\lq horizontal\rq\rq\ dihedral wedge of $a,b$ (the wedge that does not contain the $z$-vertical direction), denoted by $W_{a,b}$. See Figure \ref{fig: translating_line_in_wedge}. Moreover, the third plane $d$ must pass below $p$.

\begin{figure}[!htbp]
		\centering
  		\includegraphics[width=3in]{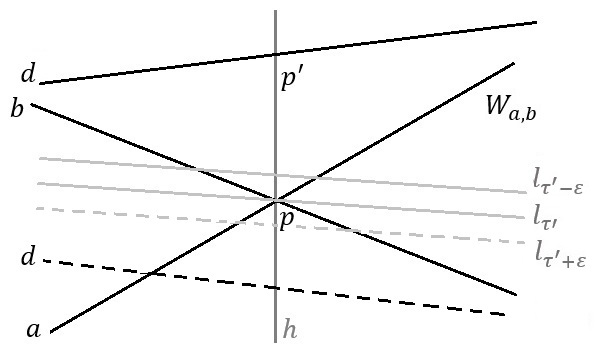}
  		\caption{\small\sf The horizontal dihedral wedge $W_{a,b}$ in which the translating line $l_{\tau}$ fully lies at $\tau=\tau'$. For small enough $\eps>0$, for each $d \in \Lambda$ that meets the $z$-vertical line through $p$ at a point $p'$ that satisfies $z_{p'}>z_p$ (resp., $z_{p'}<z_p$), we have $l_{\tau'-\eps} \subseteq C_{a,b,d}$ (resp., $l_{\tau'+\eps} \subseteq C_{a,b,d}$) (both possibilities are shown in the figure, one drawn solid, the other dashed).}  	
  	\label{fig: translating_line_in_wedge}
	\end{figure}
	
	In a similar manner, $l_{\tau'}$ is about to exit $C_{a,b,d}$ if and only if it touches an intersection line of two of these	 planes, say $a \cap b$, and the third plane $d$ passes above the contact point.
	
	Assume then that $l_{\tau'}$ touches an intersection line $a \cap b$ of two planes $a,b \in \Lambda$, and fully lies in the horizontal dihedral wedge $W_{a,b}$. Let $h$ be the vertical line through $p$. The argument just given implies that at this contact, $l_{\tau'}$ is about to enter a $k$-corridor $C_{a,b,d}$ (resp., about to exit $C_{a,b,d}$) if and only if $d \in h_{down} \cap D^k$ (resp., $d \in h_{up} \cap D^k$). Hence, the net increase in the number of containing $k$-corridors, as we pass through $a \cap b$, is $|h_{down} \cap D^k| - |h_{up} \cap D^k|$, and by Lemma \ref{lemma: antipodality}, the absolute value of this difference is at most 2. 
	There are ${\binom{n}{2}}=\frac{n(n-1)}{2}$ horizontal dihedral wedges formed by the planes of $\Lambda$.
	Thus, at the end of the process, $l_0$ is fully contained in at most $2 \cdot n(n-1)/2 = n(n-1)$ $k$-corridors. We can improve this bound by a factor of 2, noticing that either at most $\frac{1}{2} \cdot \frac{n(n-1)}{2}$ intersection lines pass above $l_0$, or at most $\frac{1}{2} \cdot \frac{n(n-1)}{2}$ intersection lines pass below $l_0$. In the former case the analysis just given yields the improved bound $\frac{n(n-1)}{2}$. In the latter case we run a symmetric version of the argument, starting with $l_1$ that lies sufficiently far downwards, and translating it upwards towards $l_0$. A suitably modified analysis yields the same bound $\frac{n(n-1)}{2}$, and the lemma follows.
	\hfill
	\end{proof}

	\begin{lemma}
	\label{lemma: upper_bound}
	The number $X^k$ of ordered pairs of $k$-corridors such that the first corridor is immersed in the second one, in the arrangement $\A(\Lambda)$, is at most $\frac{3n^4}{4}$.
	\end{lemma}

	\begin{proof}
	Fix an intersection line $l_{a,b}=a \cap b$ of two planes from $\Lambda$. by Lemma \ref{lemma: lovasz}, applied to $l_{a,b}$ and $\Lambda \setminus \{a,b \}$, $l_{a,b}$ is fully contained in at most $(n-2)(n-3)/2$ $k$-corridors. For each containing $k$-corridor $C_{c,d,e}$, $l_{a,b}$ can contribute at most three ordered pairs to $X^k$, namely an immersion of $C_{a,b,c}$ in $C_{c,d,e}$, of $C_{a,b,d}$ in $C_{c,d,e}$ and of $C_{a,b,e}$ in $C_{c,d,e}$. Since there are only ${\binom{n}{2}}$ intersection lines in $\A(\Lambda)$, we get that there are at most $3\frac{(n-2)(n-3)}{2}{\binom{n}{2}} < \frac{3n^4}{4}$ ordered pairs of immersed $k$-corridors.
	\hfill
	\end{proof}

\subsection{An attempt to generalize the dual version of the Lov\'asz Lemma}
\label{subsec: attempt_to_generalize}
	
	In this subsection we note one of the challenges in generalizing the technique, described above for the case of planes, to arrangements of more general surfaces. Assume we now have a collection of surfaces $\Lambda'$ (these are \lq almost\rq\ our pseudoplanes, which we required to also satisfy one additional property---see Section \ref{sec: pseudoplanes}), that shares similar topological properties to arrangements of planes. That is:
	\renewcommand\labelenumi{(\roman{enumi})}
	\renewcommand\theenumi\labelenumi
	\begin{enumerate}
	\item The surfaces of $\Lambda'$ are graphs of total bivariate continuous functions.
	\item The intersection of any pair of surfaces in $\Lambda'$ is a connected $x$-monotone unbounded curve.
	\item Any triple of surfaces in $\Lambda'$ intersect in exactly one point.
\end{enumerate}

	As will be shown for the case of pseudoplanes in general position (Section \ref{subsec: pseudo-lovasz_lemma}), the dual version of the
antipodality property presented for the case of planes in Lemma \ref{lemma: antipodality}, can be generalized for such a collection
$\Lambda'$. It is somewhat more intricate to generalize Lemma \ref{lemma: lovasz}. Let $\gamma_0$ be a connected $x$-monotone unbounded curve, that does not meet any of the intersection curves of pairs of surfaces in $\Lambda'$. Assume, for simplicity, that the $xy$-projection of $\gamma_0$ is a straight line. Consider the vertical plane through $\gamma_0$, and denote the intersection curves of the surfaces from $\Lambda'$ with $\gamma_0$ by $\Sigma = \{ \sigma_a = \gamma_0 \cap a \mid a\in\Lambda' \}$.
	The technique in Lemma \ref{lemma: lovasz} would suggest that we consider a curve $\gamma_1$, that is a copy of $\gamma_0$ and contained in the vertical plane through $\gamma_0$, so that $\gamma_1$ lies above all the $k$-corridors of $\Lambda'$ (where the concept of $k$-corridors is extended in a natural way; for more details, see Section \ref{sec: pseudoplanes}). For each $0 \leq \tau \leq 1$, we should then define $\gamma_{\tau} = \tau \gamma_0 + (1-\tau) \gamma_1$, so that, as $\tau$ increases from $0$ to $1$, $\gamma_{\tau}$ translates from $\gamma_1$ down to $\gamma_0$.
	
	In the technique used for the case of planes, during the translation, we count the number of $k$-corridors from $\Lambda'$ that the translating curve gets out of or gets into, at any critical event, until we reach $l_0$ at $\tau=1$. In the case of planes, the critical events are where for some $0 < \tau' < 1$, the curve (line) $\gamma_{\tau'}$ is fully contained in the \lq\lq horizontal\rq\rq\ dihedral wedge $W_{a,b}$, of some $a,b \in \Lambda'$ (the wedge that does not contain the $z$-vertical direction), and so $\gamma_{\tau}$, touches the intersection curve $a\cap b$. The translating line is about to enter to, or exit from, each of the $k$-corridors that defined by $a,b$ and some other input plane $d$.
	
	\begin{figure}[!htbp]
		\centering
  		\includegraphics[width=6in]{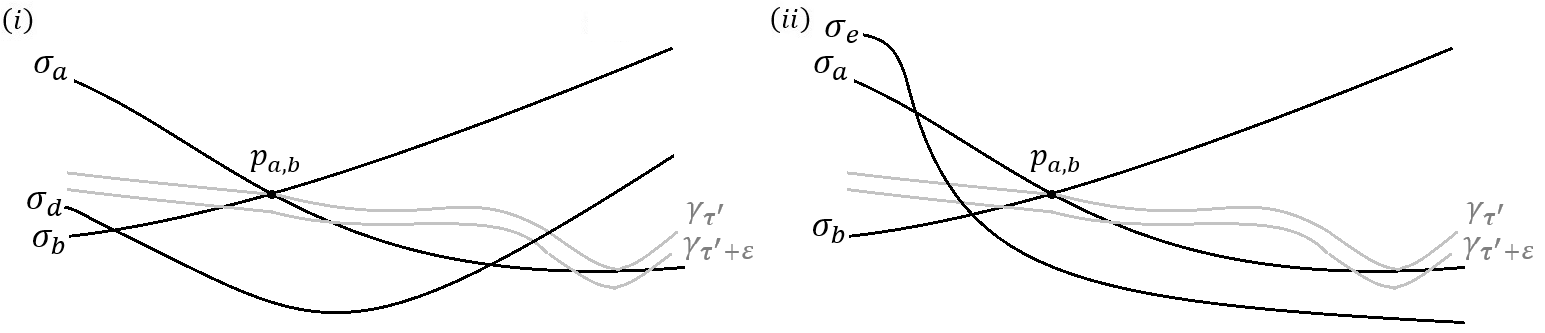}
  		\caption{\small\sf The horizontal dihedral wedge $W_{a,b}$ in which the translating curve $\gamma_{\tau'}$ fully lies at $\tau=\tau'$. The curves $\sigma_d, \sigma_e$ pass below $p_{a,b} = \sigma_a \cap \sigma_b$. (i) For any $\eps>0$, $\gamma_{\tau'+\eps}$ is not fully contained in $C_{a,b,d}$. (ii) For small enough $\eps>0$, $\gamma_{\tau'+\eps}$ is fully contained in $C_{a,b,e}$.} 
  	\label{fig: attamp_to_generalize1}
	\end{figure}

	For collections $\Lambda'$ of more general surfaces, satisfying (i)--(iii) above, these events are more complicated: locally, around the
intersection point $p_{a,b} = \sigma_a \cap \sigma_b$, the translating curve $\gamma_{\tau'}$ is about to get into or get out from the
corresponding $k$-corridor, but this does not mean that $\gamma_{\tau' + \eps}$ or $\gamma_{\tau' - \eps}$ is fully contained in that
$k$-corridor, for any small $\eps>0$. Such a situation is depicted in Figure \ref{fig: attamp_to_generalize1}(i). The curve $\gamma_{\tau'}$ is the translating curve at a moment $0 < \tau' < 1$, so that $\gamma_{\tau'}$ is fully contained in the horizontal dihedral wedge $W_{a,b}$. See Figure \ref{fig: attamp_to_generalize1}(i). For small enough $\eps>0$, the part of the curve $\gamma_{\tau'+\eps}$ with points that close enough to the intersection point $p_{a,b} = \sigma_a \cup \sigma_b$ are between the upper envelope of $\{a,b,d\}$ and the lower envelope of $\{a,b,d\}$, but for any $\eps>0$, $\gamma_{\tau'+\eps}$ is not fully contained in $C_{a,b,d}$. In Figure \ref{fig: attamp_to_generalize1}(ii), for the same $a,b,\gamma_{\tau'}$, there is a small enough $\eps'>0$, so that for every $0<\eps<\eps'$, $\gamma_{\tau'+\eps}$ is fully contained in $C_{a,b,e}$, and $\gamma_{\tau'}$ is about to enter that corridor (as for the case of planes).

	A similar situation is where $\gamma_{\tau'}$ is fully contained in $W_{a,b}$, and supposedly is about to get out from corridors $C_{a,b,d}$ where $d \in \Lambda'$ passes above the intersection point $p_{a,b}$. It is possible in the current setting that for any $\eps>0$, $\gamma_{\tau'-\eps}$ is not fully contained in $C_{a,b,d}$, and therefore $\gamma_{\tau'}$ does not get out from these corridors.
	
Note that for these problematic scenarios to arise, $\gamma_{\tau'}$ and $\gamma_{\tau'+\eps}$ have to intersect $\sigma_a$ more than once
(twice and three times, respectively). This bad behavior would have not occured if we were to require that the translating curve 
intersects each of the `static' cross-sections, like $\sigma_a$, at most once, but this is too strong an assumption to make---see 
Section~\ref{sec: pseudoplanes} for an additional discussion.

	We conclude that we cannot apply the same technique used for the case of planes, which is based on the antipodality property, in order
to generalize the dual version of the Lov\'asz Lemma, even for collections $\Lambda'$ with the above properties, which are natural ropological
generalizations of similar properties of planes. A technique that overcomes this issue,
one of the novel ingrediants of our approach, is presented in Section \ref{sec: pseudoplanes}.

\section{The case of pseudoplanes}
\label{sec: pseudoplanes}
	We say that a family $\Lambda$ of $n$ surfaces in $\mathbb{R}^3$ is a \textit{family of pseudoplanes} in general position if (observe that properties (i)--(iii) are copied from Section \ref{subsec: attempt_to_generalize})
	
	\renewcommand\labelenumi{(\roman{enumi})}
	\renewcommand\theenumi\labelenumi
	\begin{enumerate}
	\item The surfaces of $\Lambda$ are graphs of total bivariate continuous functions.
	\item The intersection of any pair of surfaces in $\Lambda$ is a connected $x$-monotone unbounded curve.
	\item Any triple of surfaces in $\Lambda$ intersect in exactly one point.
	\item The $xy$-projections of the set of all ${\binom{n}{2}}$ intersection curves of the surfaces form a \textit{family of pseudolines} in the plane. That is, this is a collection of  ${\binom{n}{2}}$ $x$-monotone unbounded curves, each pair of which intersect exactly once; see~\cite{AS02} for more details.
\end{enumerate}
	
	The assumption that the pseudoplanes of $\Lambda$ are in general position means that no point is incident to more than three pseudoplanes, no intersection curve of two pseudoplanes is tangent to a third pseudoplane, and no two pseudoplanes are tangent to each other. We note that conditions (\romannum{1})--(\romannum{3}) are natural (as was already noted in Section \ref{subsec: attempt_to_generalize}), but condition (\romannum{4}) might appear somewhat restrictive, even though it obviously holds for planes. For any $a,b,c \in \Lambda$, we denote the intersection curve $a \cap b$ by $\gamma_{a,b}$, and the intersection point $a \cap b \cap c$ by $p_{a,b,c}$.
	
	\begin{definition}
	\label{def: curtain}
	Let $\gamma$ be a curve in $\mathbb{R}^3$. The \emph{vertical curtain through $\gamma$}, denoted by \emph{$\Upsilon_{\gamma}$}, is the union of all the $z$-vertical lines that intersect $\gamma$. The portion of $\Upsilon_{\gamma}$ above (resp., below) $\gamma$ is called \emph{the upper} (resp., \emph{lower}) \emph{curtain} of $\gamma$, and is denoted by \emph{$\Upsilon_{\gamma}^{u}$} (resp., \emph{$\Upsilon_{\gamma}^{d}$}).
	\end{definition}
	
	Let $\gamma$ be an $x$-monotone unbounded connected curve in $\mathbb{R}^3$, and let $p \in \gamma$. We call each of the two connected components of $\gamma \setminus \{p\}$ a \emph{half-curve} of $\gamma$ emanating from $p$. 
		
	The following lemma is an immediate consequence of the general position of the pseudoplanes in $\Lambda$:
	
	\begin{lemma}
    \label{lemma: pseudo-properties}
    Let $a,b,c \in \Lambda$, and let $\gamma_{a,b} = a \cap b$, $p_{a,b,c} = a \cap b \cap c$.
    \ifsocg
    \begin{enumerate}[(a)]
    \else
    \begin{enumerate}
    \fi
	\item One of the two half-curves of $\gamma_{a,b}$ that emanates from $p_{a,b,c}$ lies fully below $c$, and the other half-curve lies fully above $c$.
	\item The collection of intersections between the surfaces of $\Lambda$ and $\Upsilon_{\gamma_{a,b}}$ forms an arrangement of unbounded $x$-monotone curves on $\Upsilon_{\gamma_{a,b}}$, each pair of which intersect at most once.\footnote{In a sense, this is a collection of pseudolines, except that they are, in general, not drawn in a plane.}
\end{enumerate}
	\end{lemma}
	
	\begin{proof}	
	The proof of (a) is straightforward and is omitted.
	For (b), property (\romannum{4}) implies that, for any $c,d \in \Lambda \backslash \{a,b\}$, the projection on the $xy$-plane of $\gamma_{a,b}$ and $\gamma_{c,d}=c \cap d$ intersect at most once. Thus, the intersection curves $c \cap \Upsilon_{\gamma_{a,b}}$, $d \cap \Upsilon_{\gamma_{a,b}}$ intersect at most once.
	\hfill
	\end{proof}

	Another property of $\A(\Lambda)$, shown in Agarwal and Sharir \cite{AM95}, is:
	
	\begin{lemma}
    \label{lemma: linear_envelope}
    The complexity of the lower envelope of $\Lambda$ is $O(n)$.
	\end{lemma}
	
	This lemma will be useful when we derive, in Section \ref{subsec: pseudo_bounding_level}, a $k$-sensitive bound on the complexity of the $k$-level.

	The notion of corridors can easily be extended to the case of pseudoplanes. That is, for any $a,b,c \in \Lambda$, denote by $C_{a,b,c}$ the open region between the lower envelope and the upper envelope of $a,b,c$, and call it the \textit{corridor of $a,b,c$}.
Refer to corridors $C_{a,b,c}$ for which the intersection point $p_{a,b,c}$ lies at level $k$ as \emph{$k$-corridors}, and define $C^k$ as the collection of $k$-corridors in $\A(\Lambda)$. The following is an extension of Definition \ref{def: immersed}:

	\begin{definition}
	\label{def: pseudo_immersed}
	A corridor $C_1$ \emph{\textbf{is immersed}} in a corridor $C_2$ if they share exactly one pseudoplane, and the intersection curve of the other two pseudoplanes of $C_1$ is fully contained in $C_2$. Let $X^k$ denote the number of ordered pairs of immersed corridors in $C^k$. 
	\end{definition}

\paragraph{Organization of this section.} In Section \ref{subsec: pseudo-lovasz_lemma} we derive an upper bound for $X^k$, using an extended dual version of the Lov\'asz Lemma (overcoming the technical issue noted in Section \ref{subsec: attempt_to_generalize}). In Section \ref{subsec: crossing_lemma_con't} we obtain a lower bound for $X^k$, using a dual version of the Crossing Lemma. In Section \ref{subsec: pseudo_bounding_level} we combine those two bounds to obtain an upper bound on the complexity of the $k$-level of the arrangement.

\subsection{An extension of the dual version of the Lov\'asz Lemma} 
	\label{subsec: pseudo-lovasz_lemma}
	
	The following lemma extends Lemma \ref{lemma: antipodality} to the case of pseudoplanes.
	
    \begin{lemma}
    \label{lemma: pseudo-antipodality}
	Let $\Lambda$ be a collection of pseudoplanes, as defined above, let $a, b\in\Lambda$, and let $\gamma_{a,b} = a \cap b$ denote their intersection curve. Let $h$ be a $z$-vertical line (i.e., parallel to the $z$-axis) that intersects $\gamma_{a,b}$ at some point $p$. Let $D = \Lambda \setminus \{a,b\}$, and $D^k = \{d \in D \mid C_{a,b,d} \in C^k\}$. Denote $h_{up} = \{d \in D \mid z_{d \cap h}>z_p\}$ and $h_{down} = \{d \in D \mid z_{d \cap h}<z_p\}$ (by choosing $p$ generically, we may assume that all these inequalities are indeed sharp).
	We then have $\Bigl| |h_{up} \cap D^k| - |h_{down} \cap D^k| \Bigr| \leq 2$.
	\end{lemma}

	\begin{proof}
	The proof of Lemma \ref{lemma: antipodality} applies to the case of pseudoplanes too, more or less verbatim, because all the arguments used there are topological in nature, and can be easily extended to the case of pseudoplanes, with obvious straightforward modifications (see, e.g., 
Figure \ref{fig: pseuso_neighbors_level}, which extends the obserations depicted in Figure \ref{fig: neighbors_level}).
\end{proof}

	\begin{figure}[!htbp]
	\centering
  	\includegraphics[width=6in]{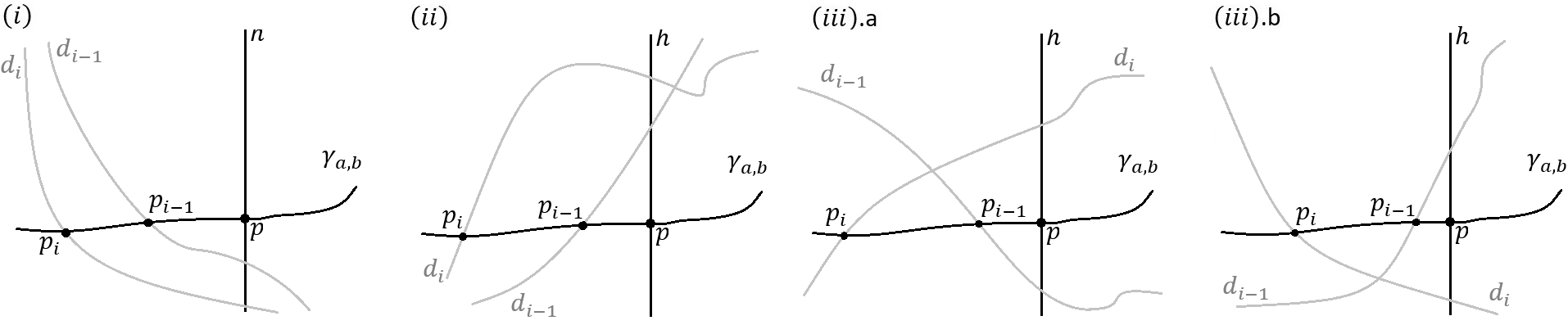}
  		\caption{\small\sf Both $d_i, d_{i-1}$ intersect the same half-curve of $\gamma_{a,b}$ emanating from $p$: (i) The case where both $d_i, d_{i-1} \in h_{down}$. (ii) The case where both $d_i, d_{i-1} \in h_{up}$. (iii).a--(iii).b The case where one of $d_i, d_{i-1}$ is in $h_{down}$, and the other one is in $h_{up}$.}  	
  		\label{fig: pseuso_neighbors_level}
	\end{figure}

	We next apply this lemma to obtain an extended dual version of the Lov\'asz Lemma. Concretely, we derive an upper bound on the number of $k$-corridors that fully contain an intersection curve $\gamma_{a,b} = a \cap b$. 
	In the case of planes (see Lemma \ref{lemma: lovasz} in Section \ref{subsec: lovasz_lemma}), $\gamma_{a,b}$ is a line, and the argument is based on a sweeping process, of a line parallel to $\gamma_{a,b}$ and lying above it, from $+ \infty$ downwards to $\gamma_{a,b}$, keeping track of the number of containing $k$-corridors as the sweeping line touches an intersection curve of two other planes.
As discussed in Section \ref{subsec: attempt_to_generalize}, it is not trivial to generalize this sweeping argument to the case of pseudoplanes.
	Instead, we use the following modified argument. First,  the sweeping is performed in the reverse order, from $\gamma_{a,b}$ upwards to a curve at $z=+\infty$. More importantly, the sweeping is no longer by translating (a copy of) $\gamma_{a,b}$, but follows the \textit{topological sweeping} paradigm of Edelsbrunner and Guibas~\cite{EG89} (see also~\cite{SH89}); the sweep curve is always fully contained in the vertical curtain $\Upsilon_{\gamma_{a,b}}$.
	
	In the context considered here, we have a collection $\Gamma$ of $n-2$ curves within $\Upsilon_{\gamma_{a,b}}$, so that each pair of
them intersects once, and we sweep the arrangement $\A(\Gamma)$ with a curve $\gamma$, so that initially (when $\gamma$ coincides with
$\gamma_{a,b}$), and at every instance during the sweep, $\gamma$ intersects every curve of $\Gamma$ at most once. The sweep is a continuous motion of $\gamma$, given as a function $\tau \to \gamma_{\tau}$, for $\tau \in \mathbb{R}^{+}$, where $\gamma_0=\gamma_{a,b}$ is the initial placement of the sweeping curve and $\gamma_{\tau}$ approaches the curve $z=+ \infty$ on $\Upsilon_{\gamma_{a,b}}$ as $\tau$ tends to $\infty$. Moreover, $\gamma_{\tau}$ lies fully below $\gamma_{\tau'}$, for $\tau < \tau'$. 
	
	The sweeping curve $\gamma$ is given an orientation, which we think of as left to right. At any time during the sweep $\gamma$ intersects some subset of the curves of $\Gamma$ in some order. This ordered sequence changes when $\gamma$ passes through a vertex of $\A(\Gamma)$ or when the set of curves intersecting $\gamma$ changes by an insertion or deletion of a curve, necessarily at the first or the last place in the sequence. (It is easily seen that $\gamma$ cannot become tangent to a curve of $\Gamma$, as this would imply a double intersection of $\gamma$ with that curve, slightly before or after the tangency, which contradicts the invariant mentioned above that we aim to maintain.) We disregard the continuous nature of the sweep, and discretize it into a sequence of discrete steps, where each step represents one of these changes. As shown in Hershberger and Snoeyink~\cite{SH89}, we have:	
	
	\begin{lemma}[Hershberger and Snoeyink~\cite{SH89}]
	\label{lemma: sweeping}
	Any planar arrangement of a set $\Gamma$ of bi-infinite curves, any pair of which intersect at most once, can be swept topologically, starting with any curve $\gamma \in \Gamma$, so that, at any time during the sweep, the sweeping curve intersects any other curve at most once.
	\end{lemma}
	
	Although in our case the sweep takes place within $\Upsilon_{\gamma_{a,b}}$, which is not a plane in general, we can flatten it in the $x$-direction into a plane, keeping the $z$-vertical direction unchanged, and thereby be able to apply the topological sweeping machinery of \cite{EG89, SH89} within this curtain. See also below.

	As just noted, the sweeping mechanism, as described in \cite{SH89}, proceeds in a sequence of discrete steps, each of which implements
one of three kinds of local moves, listed below, that allow to advance the sweeping curve past an intersection point of two curves, and to add
or remove curves from the set of curves intersected by the sweeping curve, without violating the (at most) $1$-intersection property.

	Let $c \in \Gamma$ be the sweeping curve at some point during the sweep, and denote by $\Xi(c)= \big( c_1,c_2,c_3,\ldots \big)$ the sequence of curves of $\Gamma$ that intersect $c$, sorted in the left-to-right order of their intersections along $c$. The basic steps of the sweep are of the following three types (see Figure \ref{fig: sweeping_operations}):
	
\begin{figure}[!htbp]
		\centering
  		\includegraphics[width=6in]{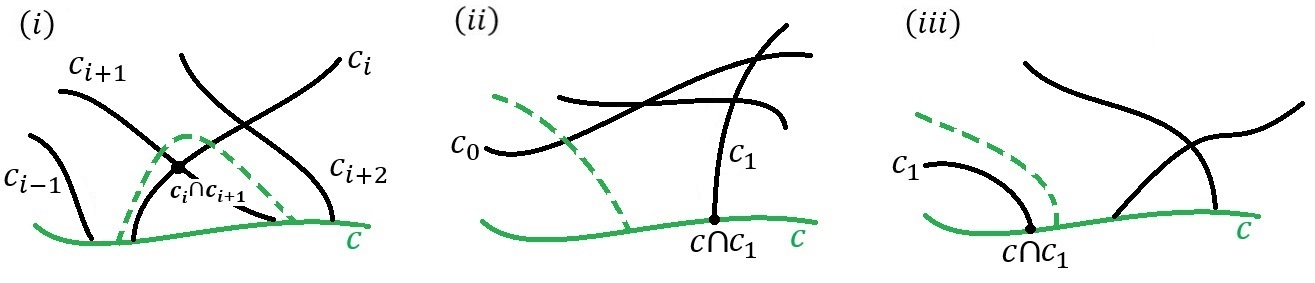}
  		\caption{\small\sf Operations by which the sweep progresses: (i) Passing over an empty triangle. (ii) Taking on the first ray. (iii) Passing over the first ray.}  	
  		\label{fig: sweeping_operations}
	\end{figure}		
	
\begin{enumerate}
	\item \textit{Passing over an empty triangle}: We have a consecutive pair of curves $c_i, c_{i+1}$ along $c$ that intersect above $c$, and no other curve passes through the (pseudo-)triangle formed by $c, c_i, c_{i+1}$. Then $c$ can move past the intersection point of $c_i,c_{i+1}$, so that the entire triangle now lies below $c$. See Figure \ref{fig: sweeping_operations}(i).
	\item \textit{Taking on the first ray}: We have a curve $c_0$ that lies above $c$ and does not intersect $c$, but $c_0$ and $c$ are adjacent on the left (i.e., at $x=-\infty$). Then we can move $c$ upwards, make it intersect $c_0$ at a point that lies to the left of all other intersection points, both on $c$ and on $c_0$. This increments the intersection sequence $\Xi(c)$ by one element, now its first element. See Figure \ref{fig: sweeping_operations}(ii).
	\item \textit{Passing over the first ray}: Here the first (leftmost) intersection point along $c$ is with a curve $c_1$ so that $c \cap c_1$ is also the first intersection along $c_1$ from the left, and $c_1$ lies above $c$ at $x = - \infty$. Then $c$ can move upwards, disentangling itself from $c_1$, and losing its intersection point with $c_1$, this time removing the first element from $\Xi(c)$. See Figure \ref{fig: sweeping_operations}(iii).
\end{enumerate}

	As shown in \cite{SH89}, we can implement the sweep so that it only performs steps of these three types, and does not have to perform the symmetric operations to (ii) and (iii), of taking on or passing over the \textit{last} (rightmost) ray.
	
	We now establish the generalized dual version of the Lov\'asz Lemma.

	\begin{lemma}
	\label{lemma: pseudo_lovasz}
	Any bi-infinite $x$-monotone curve $\gamma_0$ such that (\romannum{1}) $\gamma_0$ intersects each pseudoplane in $\Lambda$ in exactly one point, and (\romannum{2}) the $xy$-projection of $\gamma_0$ intersects the $xy$-projection of any intersection curve of two surfaces of $\Lambda$ at most once, is fully contained in at most $n(n-1)/2$ corridors in $C^k$.
	\end{lemma}

	\begin{proof}
	Let $\gamma_0$ be a curve as in the lemma, and consider the vertical curtain $\Upsilon_{\gamma_0}$ that it spans. For each pseudoplane $a \in \Lambda$, denote by $\sigma_a$ the intersection curve $a \cap \Upsilon_{\gamma_0}$, and put $\Sigma = \big\{ \sigma_a \mid a \in \Lambda \big\}$. By the assumptions of the lemma, the collection $\{ \gamma_0 \} \cup \Sigma$ of bi-infinite curves within $\Upsilon_{\gamma_0}$ has the property that any two curves in this family intersect at most once. Moreover, as already remarked earlier, by regarding the $xy$-projection $\gamma_0^*$ of $\gamma_0$ as a homeomorphic copy of the real line, we can identify $\Upsilon_{\gamma_0}$ with a homeomorphic copy of a vertical plane, where vertical lines are mapped to vertical lines. It follows that we can apply Lemma \ref{lemma: sweeping} to the arrangement of $\{\sigma_a \mid a \in \Lambda\}$ within $\Upsilon_{\gamma_0}$, and conclude that this arrangement can be topologically swept with a curve that starts at $\gamma_0$ and proceeds upwards, to infinity.

	Denote by $\gamma_{\tau}$ the sweeping curve at some moment $\tau$ during the sweeping, where the curve coincides with $\gamma_0$ at $\tau=0$.
	At the beginning of the sweep, $\gamma_0$ is fully contained in some number $Y$ of $k$-corridors, and at the end of the sweep, $\gamma_\tau$ is not contained in any of the corridors in $C^k$. We will establish an upper bound on the difference between the number of $k$-corridors that $\gamma_\tau$ gets out of (i.e., stops being fully contained in) and the number of $k$-corridors that it gets into (i.e., starts being fully contained in), at any critical event during the sweep. Summing these differences will yield the asserted upper bound on $Y$.
(As in the case of planes, the factor $1/2$ in the bound will be obtained by performing the topological sweep
either upwards or downwards, depending on which of the half-curtains $\Upsilon_{\gamma_0}^+$, $\Upsilon_{\gamma_0}^-$
contains fewer intersection points.)

	Consider $\gamma_{\tau}$ at some instance $\tau$ during the sweep, and let $a \in \Lambda$. If $\sigma_a$ is fully above $\gamma_{\tau}$, we get that $\gamma_0$, which is obtained by some motion of $\gamma_{\tau}$ downwards, is fully below $a$, a contradiction to the assumption that $\gamma_0$ intersects all the pseudoplanes in $\Lambda$. Therefore, each pseudoplane in $\Lambda$ is either fully below $\gamma_{\tau}$, or intersects it (exactly once). Hence, during the sweeping from $\gamma_0$, the only valid sweeping steps are passing over an empty triangle and passing over the first ray.

	Clearly, it suffices to consider what happens at instances $\tau$ at which $\gamma$ is about to pass through a vertex of the arrangement of $\{ \sigma_a \mid a \in \Lambda \}$ on $\Upsilon_{\gamma_0}$, or at instances at which $\gamma_{\tau}$ is about to pass over the first ray. So let $\tau^-$ and $\tau^+$ denote instances immediately before and after such a critical transition. We distinguish between three types of sweeping steps.
	\medskip 
	
	\begin{figure}[!htbp]
		\centering
  		\includegraphics[width=5in]{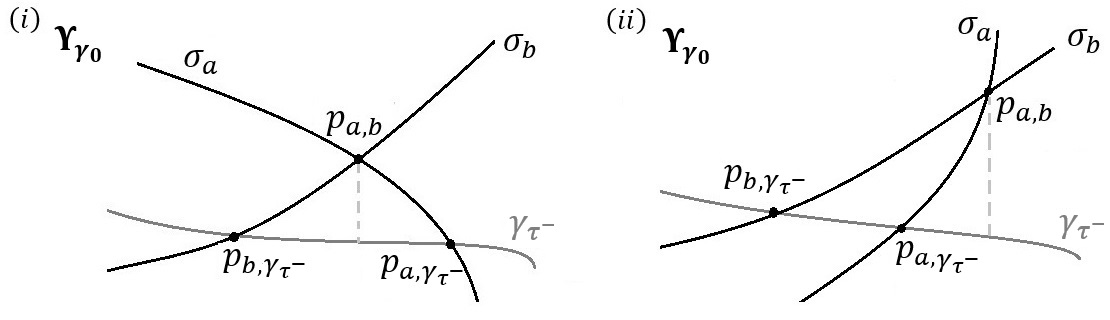}
  		\caption{\small\sf (i) The intersection point $p_{a,b}=\sigma_a \cap \sigma_b$ is directly above the curve $\gamma_{\tau^-}$ somewhere between $p_{a, \gamma_{\tau^-}}=\sigma_a \cap \gamma_{\tau^-}$ and $p_{b, \gamma_{\tau^-}}=\sigma_b \cap \gamma_{\tau^-}$. (ii) The intersection point $p_{a,b}$ is not directly above any point on the curve $\gamma_{\tau^-}$ between $p_{a, \gamma_{\tau^-}}$ and $p_{b, \gamma_{\tau^-}}$.}  
  		\label{fig: sweeping_cases}
	\end{figure}	
	
	\noindent {\bf Case 1:} The transition at $\tau$ is that we pass over an empty triangle, defined by some pair of curves $\sigma_a, \sigma_b$ and $\gamma_{\tau^-}$, such that the point on $\gamma_{\tau^-}$ directly below the intersection point $p_{a,b}$ is somewhere between 
  $p_{a, \gamma_{\tau^-}}=\sigma_a \cap \gamma_{\tau^-}$ and $p_{b, \gamma_{\tau^-}}=\sigma_b \cap \gamma_{\tau^-}$
(see Figure \ref{fig: sweeping_cases}(i)). Since $\gamma_{\tau}$ intersects each of $\sigma_a, \sigma_b$ at most once, almost all of the curve $\gamma_{\tau^-}$ is between the lower envelope and the upper envelope of $\{ \sigma_a, \sigma_b \}$, except for its portion between $p_{a, \gamma_{\tau^-}}$ and $p_{b, \gamma_{\tau^-}}$. Since the triangle defined by $\sigma_a, \sigma_b$ and $\gamma_{\tau^-}$ is empty, each curve in $\Sigma$ that lies below $p_{a,b}$ defines a corridor with $\sigma_a, \sigma_b$, such that $\gamma_{\tau^-}$ lies fully in that corridor. Symmetrically, $\gamma_{\tau^+}$, for $\tau^+$ sufficiently close to $\tau$, lies fully in each corridor defined by $\sigma_a, \sigma_b$ and a curve in $\Sigma$ that passes above $p_{a,b}$ (see, e.g., Figure \ref{fig: sweeping_line_3_case}(i)). Hence, by Lemma \ref{lemma: pseudo-antipodality}, the absolute value of the difference between the number of $k$-corridors that $\gamma_{\tau}$ gets out of at $\tau$, and the number of $k$-corridors that $\gamma_{\tau}$ gets into, is at most $2$.
	
	\begin{figure}[!htbp]
		\centering
  		\includegraphics[width=6in]{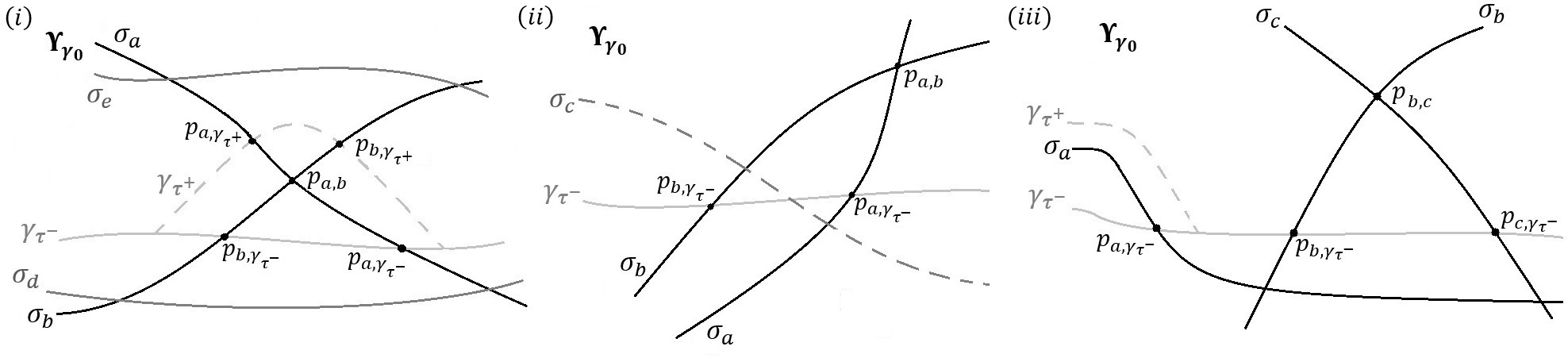}
  		\caption{\small\sf The three cases of critical events during the sweep.}  
  		\label{fig: sweeping_line_3_case}
	\end{figure}	
	
	\medskip
	\noindent {\bf Case 2:} The transition at $\tau$ is that we pass over an empty triangle, defined by $\sigma_a, \sigma_b$ and $\gamma_{\tau^-}$, where now the vertical projection of the intersection point $p_{a,b}$ onto $\gamma_{\tau^-}$ is not between $p_{a, \gamma_{\tau^-}}$ and $p_{b, \gamma_{\tau^-}}$, but lies on one side of both, say past $p_{a, \gamma_{\tau^-}}$ to the right (see Figure \ref{fig: sweeping_cases}(ii)). We claim that neither $\gamma_{\tau^-}$ nor $\gamma_{\tau^+}$ is fully contained in any corridor $C_{a,b,c}$, for $\sigma_c \in \Sigma$.
	Indeed (refer again to Figure \ref{fig: sweeping_cases}(ii)), by Lemma \ref{lemma: pseudo-properties}, the half-curve emanating from $p_{a, \gamma_{\tau^-}}$ on $\gamma_{\tau^-}$ to the right is below the lower envelope of $\{\sigma_a, \sigma_b\}$, and the half-curve emanating from $p_{b, \gamma_{\tau^-}}$ on $\gamma_{\tau^-}$ to the left is above the upper envelope of $\{\sigma_a, \sigma_b\}$. Hence, in order for $\gamma_{\tau^-}$ to be fully contained in a corridor $C_{a,b,c}$ for some other curve $\sigma_c \in \Sigma$, $\sigma_c$ must pass above $p_{b, \gamma_{\tau^-}}$ and below $p_{a, \gamma_{\tau^-}}$, and therefore it must intersect the triangle defined by $\sigma_a$, $\sigma_b$ and $\gamma_{\tau^-}$ (see Figure \ref{fig: sweeping_line_3_case}(ii)). Since this triangle is empty, there is no such $\sigma_c$. Symmetrically\footnote{Indeed, right after the transition, $\sigma_a, \sigma_b$ and $\gamma_{\tau^+}$ form an empty triangle above $p_{a,b}$, with similar properties that allow us to apply a symmetric variant of the argument just presented.}, $\gamma_{\tau^+}$ is not contained in any corridor $C_{a,b,c}$ for any $c$. 
	Hence, at this step in the sweeping process, there is no change in the set of corridors that fully contain $\gamma_{\tau}$.
	
	\medskip
	\noindent {\bf Case 3:} The transition at $\tau$ is passing over the first ray, belonging to some $\sigma_a \in \Sigma$. We claim that here too $\gamma_{\tau^-}$ and $\gamma_{\tau^+}$ are fully contained in the same corridors.	
	Indeed, except for the left ray of $\sigma_a$, $\gamma_{\tau^-}$ and $\gamma_{\tau^+}$ are fully above $\sigma_a$. Moreover, the only corridors that $\gamma$ can get into or out of at this transition must involve $\sigma_a$. Let $C_{a,b,c}$ be such a corridor. If $\sigma_a$ appears on both the upper and the lower envelopes of $\{ \sigma_a, \sigma_b, \sigma_c \}$ then, as is easily checked, neither $\gamma_{\tau^-}$ nor $\gamma_{\tau^+}$ can be fully contained in $C_{a,b,c}$. Hence, $\sigma_a$ must appear on exactly one of the envelopes, and then it must appear there as the middle portion of the envelope (see Figure \ref{fig: sweeping_line_3_case}(iii)).\footnote{If $\sigma_a$ appears as the leftmost or rightmost portion of one of the envelopes it must also appear on the other envelope.} But then the left ray of $\sigma_a$ over which $\gamma$ is swept cannot appear on either envelope, so the transition does not cause $\gamma$ to enter or leave $C_{a,b,c}$, as claimed.

	In summary, the only kind of step during the sweep in which the number of $k$-corridors that the sweeping curve is contained in changes,
is passing over an empty triangle with the structure considered in Case 1, and then, as argued above, this number can change by at most $2$.
	There are ${\binom{n}{2}}=\frac{n(n-1)}{2}$ intersection points on $\Upsilon_{\gamma_0}$, since each pair of curves in $\Sigma$ intersects at most once. We may assume without loss of generality that at most half of them are above $\gamma_0$: in the complementary case, 
as already remarked, we reverse the direction of the topological sweep, sweeping from $\gamma_0$ downwards, towards $- \infty$.
In either case, the sweeping curve can encounter at most $\frac{1}{2} \cdot \frac{n(n-1)}{2}$ empty triangles whose middle vertex lies above the
opposite edge (as in Case 1). Thus, at the beginning of the process, $\gamma_0$ is fully contained in at most $\frac{n(n-1)}{2}$ $k$-corridors
and the lemma follows.
	\hfill
	\end{proof}

	As a corollary (compare with Lemma \ref{lemma: upper_bound}, for the case of planes), we obtain the following upper bound on $X^k$: 
	
	\begin{lemma}
	\label{lemma: pseudo_upper_bound}
	The number $X^k$ of ordered pairs of $k$-corridors such that the first corridor is immersed in the second one, in the arrangement $\A(\Lambda)$, is at most $\frac{3n^4}{4}$.
	\end{lemma}

	\begin{proof}
	Fix an intersection curve $\gamma_{a,b}=a \cap b$ of two pseudoplanes from $\Lambda$. By Lemma \ref{lemma: pseudo_lovasz}, applied
within $\Upsilon_{\gamma_{a,b}}$ to the collection $\Lambda \setminus \{a,b\}$, $\gamma_{a,b}$ is fully contained in at most 
$\frac{(n-2)(n-3)}{2}$ $k$-corridors. For each containing $k$-corridor $C_{c,d,e}$, $\gamma_{a,b}$ can contribute at most three ordered pairs to $X^k$, namely an immersion of $C_{a,b,c}$ in $C_{c,d,e}$, of $C_{a,b,d}$ in $C_{c,d,e}$, and of $C_{a,b,e}$ in $C_{c,d,e}$. Since there are only ${\binom{n}{2}}$ intersection curves in $\A(\Lambda)$, we get that there are at most $3\frac{(n-2)(n-3)}{2} {\binom{n}{2}} < \frac{3n^4}{4}$ ordered pairs of immersed $k$-corridors.
	\hfill
	\end{proof}

\subsection{The dual version of the Crossing Lemma}
\label{subsec: crossing_lemma_con't}

	In this subsection we derive a lower bound on $X^k$, using a dual version of the Crossing Lemma (see~\cite{ACNS82}), extended to the
case of pseudoplanes. For each pseudoplane $a \in \Lambda$, denote by $z_a$ the intersection point of $a$ with the $z$-axis. We can choose the
position of the $z$-axis so as to ensure that (a) all the values $z_a$ are distinct, and (b) for each $a \in \Lambda$, $z_a$ lies above (in the $y$-direction of the $xy$-projection of $a$) all the intersection curves $\gamma_{a,b}$, for $b \in \Lambda \setminus \{a\}$.
	
	\begin{definition}
	\label{def: pseudo_Gamma_a}
	Let $a \in \Lambda$. Denote by $\Gamma_a$ the collection of the intersection curves of $a$ and the other pseudoplanes $b \in \Lambda$ with $z_b > z_a$. That is, $\Gamma_a=\{\gamma_b := a \cap b \mid b \in \Lambda \setminus \{a\}, z_b > z_a\}$. 
	\end{definition}

	By the assumptions on $\Lambda$, the $xy$-projection of any intersection curve of two pseudoplanes in $\Lambda$ is an $x$-monotone curve. Therefore, $\Gamma_a$ forms a family of $x$-monotone curves on the surface $a$. Since $a$ is the graph of a bivariate continuous function, it will be convenient to identify it with its $xy$-projection, and think of it, for the purpose of the current analysis, as a horizontal plane. Each pair of curves from $\Gamma_a$ intersects exactly once, because each triple of pseudoplanes in $\Lambda$ intersects exactly once. Each curve in $\Gamma_a$ is bi-infinite and divides $a$ into two unbounded regions. These considerations allow us to interpret $\Gamma_a$ as a family of $x$-monotone pseudolines in the plane.

	\begin{definition}
	\label{def: pseudo_horizontal_wedge}
	Let $a \in \Lambda$, and let $\Gamma_a$ be as above. Each $d \in \Lambda \setminus \{a\}$ for which $\gamma_d \in \Gamma_a$ divides $a$ into two disjoint regions: the region $a_d^-$ on $a$ that is fully above the pseudoplane $d$, and the region $a_d^+$ on $a$ that is fully below $d$ (so $a_d^-$ means that $d$ is below $a$, and $a_d^+$ means that $d$ is above $a$). These two regions are delimited by the intersecion curve $\gamma_d$ on $a$. Note that $z_a \in a_d^+$. That is, $a^{+}_{d}$ is the region that lies above (in the $y$-direction) the intersection curve $\gamma_d$, and $a^{-}_{d}$ is the region below $\gamma_d$.
	
	For each pair of distinct pseudoplanes $b, c \in \Lambda \setminus \{a\}$ such that $\gamma_b, \gamma_c \in \Gamma_a$, define the 
\emph{\textbf{$x$-horizontal  wedge}} \emph{\boldmath{$W^a_{b,c}$}} as the region on the pseudoplane $a$ that is contained in exactly one of the two regions $a_b^+, a_c^+$, that is, in exactly one of the regions that are bounded by $\gamma_b, \gamma_c$ and contain $z_a$ (see Figure \ref{fig: pseudo_horizontal_wedge2}(1)).
	\end{definition}

	Note that our assumption on the position of $z_a$ in $a$ allows us to regard the wedges $W^a_{b,c}$ as being indeed \lq $x$-horizontal\rq\ with respect to the $xy$-frame in the $xy$-projection of $a$. 
	
	Continue to fix the pseudoplane $a$, let $E_a$ be some subset of vertices of $\A(\Gamma_a)$, and let $G_a=(\Gamma_a, E_a)$ denote the graph whose vertices are the pseudolines in $\Gamma_a$ and whose edges are the pairs that form the vertices of $E_a$.
		A \textbf{\textit{diamond}} in $G_a$ is two pairs $\{\gamma_b, \gamma_c\}, \{\gamma_d, \gamma_e\}$ of curves of $\Gamma_a$ on $a$, both pairs belonging to $E_a$, with all four pseudoplanes $b,c,d,e$ distinct, such that $p_{a,b,c} = \gamma_b \cap \gamma_c \in W^a_{d,e}$ and $p_{a,d,e} = \gamma_d \cap \gamma_e \in W^a_{b,c}$. See Figure \ref{fig: pseudo_horizontal_wedge2}(2.(i)) and \ref{fig: pseudo_horizontal_wedge2}(2.(ii)).

\begin{figure}[!htbp]
		\centering
  		\includegraphics[width=6in]{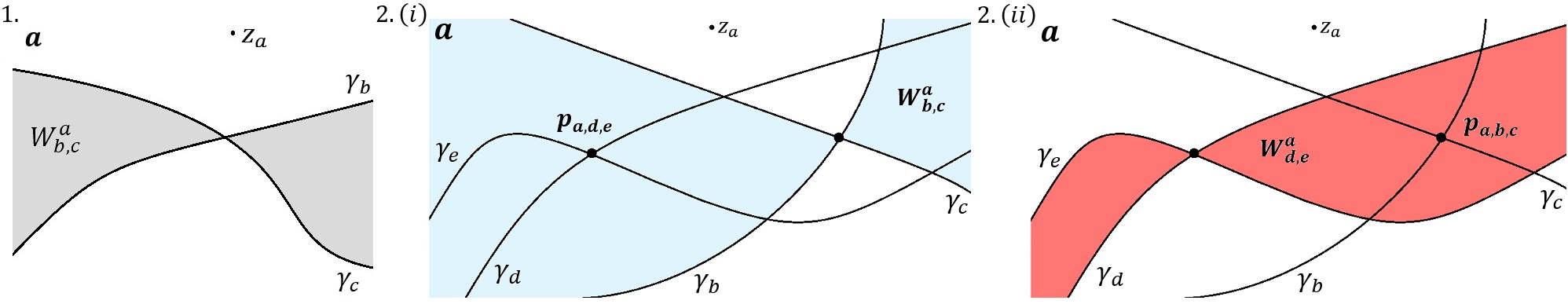}
  		\caption{\small\sf 1. The $x$-horizontal wedge $W^a_{b,c}$ on the pseudoplane $a$. The regions $a_b^+$ and $a_c^+$ lie above (in the $y$-direction) the respective curves $\gamma_b$ and $\gamma_c$, and they both contain $z_a$. 2. The two pairs $\{ \gamma_b, \gamma_c \}, \{ \gamma_d, \gamma_e \}$ on $a$ form a diamond in $G_a$. (i) The blue area is the $x$-horizontal wedge $W^a_{b,c}$, and it contains $p_{a,d,e} = \gamma_d \cap \gamma_e$. (ii) The red area is the $x$-horizontal wedge $W^a_{d,e}$, and it contains $p_{a,b,c} = \gamma_b \cap \gamma_c$.}  	
  		\label{fig: pseudo_horizontal_wedge2}
	\end{figure}	
	
	The following is our version of an extension of the dual version of Euler's formula for planar maps, derived in Tamaki and Tokuyama~\cite{TT03}, for the case of pseudolines in $\mathbb{R}^2$:
	
	 \begin{lemma}[A version of Tamaki and Tokuyama~\cite{TT03}]
	\label{lemma: pseudo_palanar_graph}
	For a pseudoplane $a \in \Lambda$, let $G_a$ be as defined above, with $|\Gamma_a| >3$. If $\Gamma_a$ is diamond-free, then $G_a$ is planar, and so $|E_a| \leq 3|\Gamma_a|-6$.
	\end{lemma}
	
	As a corollary of the lemma, we obtain the following generalized dual version of the Crossing Lemma (see~\cite{ACNS82}). For completeness, we include the proof, with suitable adjustments to accommodate the duality and the generalization to arrangement of pseudolines.

	\begin{lemma}[Generalized dual version of the Crossing Lemma]
	\label{lemma: pseudo_diamonds}
	Let $\Gamma_a$ and $G_a$ be as above, so that $|E_a|>4|\Gamma_a|$. The number of diamonds in $G_a$ is at least $\frac{|E_a|^3}{64|\Gamma_a|^2}$. 
	\end{lemma}
	
	\begin{proof}
	Let $\Gamma$ be a family of pseudolines in the (standard) plane, and $E$ a subset of the vertices of $\A(\Gamma)$. Denote by $\Delta$ the number of diamonds in $G=(\Gamma,E)$ (where the wedges used to define the diamonds are $x$-horizontal wedges, relative to some point in the plane which lies above all the curves in $\Gamma$, in the $y$-direction, similar to the way $x$-horizontal wedges were defined in Definition \ref{def: pseudo_horizontal_wedge} relative to $z_a$). We now repeat the following process, until there are no diamonds left in $G$: For a surviving diamond, remove from $E$ one of the two points that form the diamond. This eliminates the diamond and maybe some other diamonds too. Continue the process with the new $G$. The dual version of Euler's formula implies that if $|E| > 3|\Gamma|-6$ then there is a diamond in $G$. We stop the process after removing at most $|E|-3|\Gamma|+6$ vertices, each time removing at least one diamond from $G$. Hence, for the original graph $G$, we have $\Delta \geq |E|-3|\Gamma|$.

	Denote by $\delta_a$ the number of diamonds in $G_a$. Consider a random subgraph of $G_a$, in which each vertex (which is a pseudoline in $\Gamma_a$) is chosen independently with the same probability $p$. The expected number of vertices, edges and diamonds in the induced subgraph of $G_a$ is $pm$, $p^2|E_a|$, and $p^4 \delta_a$, respectively. Using linearity of expectation, we have $p^4 \delta_a>p^2|E_a|-3pm$, which implies that $\delta_a>\frac{|E_a|}{p^2} - \frac{3m}{p^3}$. For $p=\frac{4m}{|E_a|}$ (note that, by assumption, $|E_a|>4m$, so $p<1$), we obtain that $\delta_a>\frac{|E_a|^3}{64m^2}$.
	\hfill
	\end{proof}
	

	We now specialize this result to our context. For each $a \in \Lambda$, consider the set $E_a^k=\{p_{a,b,c}=\gamma_b \cap \gamma_c \mid \gamma_b, \gamma_c \in \Gamma_a, C_{a,b,c} \in C^k\}$, and the graph $G_a^k=(\Gamma_a, E_a^k)$ defined as above. Lemma \ref{lemma: pseudo_diamonds} implies the following:

	\begin{lemma}
	\label{lemma: pseudo_lower_bound}
	The number $X^k$ of ordered pairs of immersed $k$-corridors in the arrangement $\A(\Lambda)$ is at least $\frac{|C^k|^3}{64n^4} - n^2$.	
	\end{lemma}

\begin{proof}
	Let $a \in \Lambda$, $G_a^k=(\Gamma_a, E_a^k)$ be as above, and define $\Delta_a$ as the number of diamonds in $G_a^k$. Let $\{\gamma_b, \gamma_c\}, \{\gamma_d, \gamma_e\}$ be a pair that form a diamond. Since the pseudoplanes in $\Lambda$ satisfy property (\romannum{4}) and are in general position, the $xy$-projections of the curves $\gamma_{b,c}$ and $\gamma_{d,e}$ have exactly one intersection point (but $\gamma_{b,c}$ and $\gamma_{d,e}$ do not intersect in $3$-space). Moreover, since $b,c$ are the graphs of total bivariate functions, the projection of their intersection curve $\gamma_{b,c}$ on $a$ is fully contained in the region $\{ p_{a,b,c} \} \cup \big\{ a^{+}_{b} \cap a^{+}_{c} \big\} \cup \big\{ a^{-}_{b} \cap a^{-}_{c} \big\}$. That is, the projection of $\gamma_{b,c}$ is disjoint from the interior of $W^a_{b,c}$. Similarly, the projection of the intersection curve $\gamma_{d,e}$ on $a$ is fully contained in the region $\{ p_{a,d,e} \} \cup \big\{ a^{+}_{d} \cap a^{+}_{e} \big\} \cup \big\{ a^{-}_{d} \cap a^{-}_{e} \big\}$, and is disjoint from the interior of $W^a_{d,e}$. In addition, the portion of $\gamma_{b,c}$ that projects to $a^{+}_{b} \cap a^{+}_{c}$ lies above $a$ and the portion projecting to $a^{-}_{b} \cap a^{-}_{c}$ lies below $a$. A similar property holds for $\gamma_{d,e}$.

	Assume without loss of generality that the pseudoplanes $b,c,d,e$ intersect $a$ as in Figure \ref{fig: diamond_intersection_property2}(i); that is, $p_{a,b,c}$ is contained in $a_d^- \cap a_e^+$ and $p_{a,d,e}$ is contained in $a_b^- \cap a_c^+$. Since $\{ \gamma_b, \gamma_c \}$ and $\{ \gamma_d, \gamma_e \}$ form a diamond and each pair of curves on $a$ intersects exactly once, the intersection of the boundary of $a_d^+ \cap a_e^+$ and the boundary of $a_b^- \cap a_c^-$ is empty. Indeed, the interior of the arc $p_{a,b,c}p_{a,d,e}$ is fully contained in $a_b^- \cap a_c^+$, and the half-curve of $\gamma_d$ emanating from $p_{a,b,d}$ and not containing $p_{a,d,e}$, is fully contained in $a_b^+$ (otherwise, $\gamma_d$ would intersect $\gamma_b$ more than once). On the other hand, the half-curve of $\gamma_e$ emanating from $p_{a,d,e}$ and not containing $p_{a,c,e}$, is fully contained in $a_c^+$, since $\gamma_c$ already intersects the other half-curve of $\gamma_e$. These two observations establish our claim. The regions $a^{-}_{b} \cap a^{-}_{c}$ and $a^{+}_{d} \cap a^{+}_{e}$ are not contained in one another, and therefore their intersection is empty. Similarly, The intersection of $a^{+}_{b} \cap a^{+}_{c}$ and $a^{-}_{d} \cap a^{-}_{e}$ is empty.
		
	\begin{figure}[!htbp]
		\centering
  		\includegraphics[width=6in]{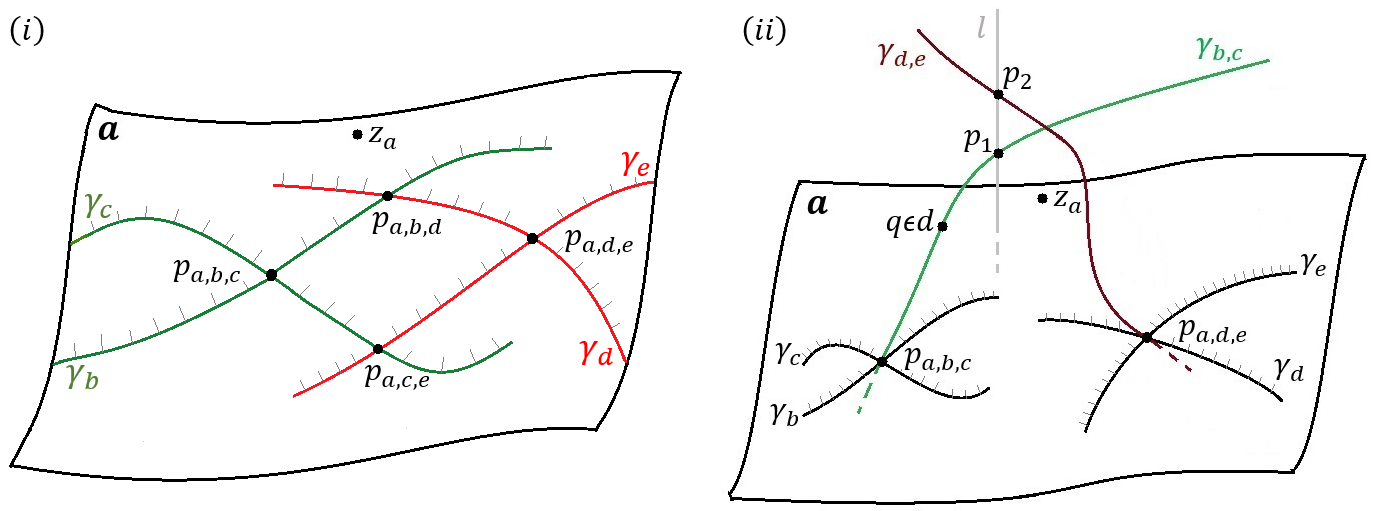}
  		\caption{\small\sf The two pairs $\{\gamma_b, \gamma_c\}, \{\gamma_d, \gamma_e\}$ on $a$ form a diamond in $G_a^k$. (i) The diamond, where $p_{a,b,c}$ is contained in $a_d^- \cap a_e^+$ and $p_{a,d,e}$ is contained in $a_b^- \cap a_c^+$. (ii) The intersection curve $\gamma_{b,c}$ is below the intersection curve $\gamma_{d,e}$. The pseudoplane $d$ (resp., $e$) meets $\gamma_{b,c}$ at a point $q$ between $p_{a,b,c}$ and $p_1$ (resp., at a point, not drawn, outside this arc).}
  		\label{fig: diamond_intersection_property2}  	
	\end{figure}

	 Assume without loss of generality that $\gamma_{b,c}$ passes below $\gamma_{d,e}$. That is, letting $l$ denote the unique $z$-vertical line that meets both $\gamma_{b,c}, \gamma_{d,e}$, the points $p_1=l \cap \gamma_{b,c}$, $p_2=l \cap \gamma_{d,e}$ satisfy $z_{p_1}<z_{p_2}$. Assume without loss of generality that $l$ intersects $a$ in the region on $a$ that is the intersection of $a^{+}_{b}, a^{+}_{c}, a^{+}_{d}, a^{+}_{e}$ (the regions on $a$ induced by the curves $\gamma_b, \gamma_c, \gamma_d, \gamma_e$ and containing $z_a$). The case where $l$ intersects $a$ in the region on $a$ that is the intersection of $a^{-}_{b}, a^{-}_{c}, a^{-}_{d}, a^{-}_{e}$, is handled symmetrically. These are the only two possibilities, since the intersection of $a^{+}_{b}, a^{+}_{c}, a^{-}_{d}, a^{-}_{e}$ is empty, and so is the intersection of $a^{-}_{b}, a^{-}_{c}, a^{+}_{d}, a^{+}_{e}$.

	Since $\gamma_{d,e}$ is above $p_1$, it follows that both $d$ and $e$ themselves are above $p_1$. Moreover, since $p_{a,b,c}$ lies in $a^{-}_{d}$, $d$ must lie below $p_{a,b,c}$. Hence $d$ must intersect $\gamma_{b,c}$ at some point $q$ between $p_{a,b,c}$ and $p_1$. Moreover, since $e$ satisfies $z_e>z_a$ and $p_{a,b,c}$ lies in $a^{+}_{e}$, as in Figure \ref{fig: diamond_intersection_property2}(ii), $e$ is above $p_{a,b,c}$. Since $e$ is also above $p_1$ and $e$ is the graph of a bivariate continuous function, its single intersection point with $\gamma_{b,c}$ must be outside the arc $p_{a,b,c}p_1$ of $\gamma_{b,c}$.

	We claim that $\gamma_{b,c} \subseteq C_{a,d,e}$. Indeed, $\gamma_{b,c}$ is fully above the lower envelope of $\{a,d,e\}$: the half-curve of $\gamma_{b,c}$ that emanates from $p_{a,b,c}$ and contains $p_1$ lies above the pseudoplane $a$ ($\gamma_{b,c}$ intersects $a$ at $p_{a,b,c}$), and the complementary half-curve lies above $d$, because the intersection point $q$ of $\gamma_{b,c}$ and $d$ lies between $p_{a,b,c}$ and $p_1$.	
	The intersection curve $\gamma_{b,c}$ also lies fully below the upper envelope of $\{a,d,e\}$. That is because (i) the half-curve of $\gamma_{b,c}$ that emanates from the intersection point $q$ of $\gamma_{b,c}$ and $d$, and contains $p_1$, lies below the pseudoplane $d$, since $p_2 \in d$ is higher than $p_1$; (ii) the half-curve of $\gamma_{b,c}$ that emanates from $p_{a,b,c}$ and does not contain $p_1$, lies below the pseudoplane $a$ (again, $\gamma_{b,c}$ intersects $a$ at $p_{a,b,c}$); and (iii) the arc $p_{a,b,c}q$ is below $e$, since $e$ is above both $p_{a,b,c}, p_1$ and therefore must be above the complete arc $p_{a,b,c}p_1$, and in particular $e$ is above the smaller arc $p_{a,b,c}q$. The other cases behave similarly and lead to similar conclusions.
	
	Thus for each pair $\{\gamma_b, \gamma_c\}, \{\gamma_d, \gamma_e\}$ that form a diamond, either $\gamma_{b,c} \subseteq C_{a,d,e}$, or $\gamma_{d,e} \subseteq C_{a,b,c}$. Either way, one of the corridors $C_{a,b,c}, C_{a,d,e}$ is immersed in the other one. Notice that every diamond in $\{G_a^k\}_{a \in \Lambda}$ yields a distinct ordered pair of immersed $k$-corridors, because for each $k$-corridor $C_{a,b,c}$, the intersection point $p_{a,b,c} = a \cap b \cap c$ represents an edge of only the graph associated with the pseudoplane with the lowest intersection point with the $z$-axis.
	Hence, by the dual version of the Crossing Lemma, namely Lemma \ref{lemma: pseudo_diamonds}, we have $X^k \geq \sum\limits_{a} \frac{|E_a^k|^3}{64|\Gamma_a|^2} $, where the sum is over all those $a$ for which $|E_a^k| \geq 4|\Gamma_a|$. Any other pseudoplane $a$ satisfies $\frac{|E_a^k|^3}{64|\Gamma_a|^2} \leq |\Gamma_a|$, which implies the somewhat weaker lower bound
	
\begin{equation} \label{eq:lower_bound}
X^k \geq \sum\limits_{a \in \Lambda} \Big( \frac{|E_a^k|^3}{64|\Gamma_a|^2} - |\Gamma_a| \Big).
\end{equation}
	
	By the definition of $G_a^k$, and as just noted, each $k$-corridor $C_{a,b,c}$ in $\A(\Lambda)$ appears in exactly one of $E_a^k, E_b^k, E_c^k$ (in the graph of the pseudoplane that intersects the $z$-axis at the lowest point among the three). Thus, $\sum\limits_{a \in \Lambda}|E_a^k|=|C^k|$.
	The number of curves in $\Gamma_a$ is at most $n-1$.
	Therefore, using \eqref{eq:lower_bound} and H\"older's inequality, we get the lower bound 
	\[ X^k \geq \sum\limits_{a \in \Lambda}\Big( \frac{|E_a^k|^3}{64|\Gamma_a|^2} - |\Gamma_a| \Big) \geq \frac{1}{64n^2} \sum\limits_{a \in \Lambda}|E_a^k|^3 - n^2 \geq \frac{1}{64n^2} \cdot \frac{\left( \sum\limits_{a \in \Lambda}|E_a^k| \right) ^3}{n^2} - n^2 = \frac{|C^k|^3}{64n^4} - n^2. \]
	\hfill
	\end{proof}

\subsection{The complexity of the \texorpdfstring{$k$-level of $\A(\Lambda)$}{k-level of A(Lambda)}}
\label{subsec: pseudo_bounding_level}

	We are now ready to obtain the upper bound on the complexity of the $k$-level of $\A(\Lambda)$.

	\begin{lemma}
	\label{lemma: pseudo_k_level}
	The complexity of the $k$-level of $\A(\Lambda)$ is $O(n^{8/3})$.
	\end{lemma}
	
	\begin{proof}
	Comparing the upper bound in Lemma \ref{lemma: pseudo_upper_bound} and the lower bound in Lemma \ref{lemma: pseudo_lower_bound} for the number $X^k$ of ordered pairs of immersed $k$-corridors in $\A(\Lambda)$, we get:
	\[ \frac{3n^4}{4} \geq X^k \geq 
	 \frac{|C^k|^3}{64n^4} - n^2. \]
		
	Hence we get that $|C^k|^3 \leq 48n^8 + 64n^6$, which implies that $|C^k|=O(n^{8/3})$.
	The number of $k$-corridors is the number of vertices of $\A(\Lambda)$ at level $k$, which implies that the complexity of the $k$-level of $\A(\Lambda)$ is $O(n^{8/3})$.
	\hfill	
	\end{proof}

	Combining the upper bound in Lemma \ref{lemma: pseudo_k_level} with the general technique of Agarwal et al. \cite{AACS98}, we get the
following $k$-sensitive result.
	
	\begin{theorem}
	\label{theom: pseudo_k_level_sensitive}
	The complexity of the $k$-level of $\A(\Lambda)$ is $O(nk^{5/3})$.	
	\end{theorem}
	
	\begin{proof}
	Take a random sample $R \subseteq \Lambda$ of size $r = \lfloor n / 2k \rfloor$. The region beneath the lower envelope of $R$, denoted by $LE(R)$, can be decomposed into $O(r)$ vertical pseudo-prisms of constant complexity. For this, take the minimization diagram (i.e., the $xy$-projection of the lower envelope) of $R$ and construct its vertical decomposition (see~\cite{AM95}). That is, cut each face of the minimization diagram of $R$ into $y$-vertical pseudo-trapezoids by drawing two vertical extensions from every vertex in the diagram, one extension going upwards and one going downwards. The extensions stop when they meet another curve of the diagram or all the way to infinity. We now take each pseudo-trapezoid $\tau$ and create from it a semi-unbounded vertical prism, consisting of all the points that lie vertically below $\tau$ or on $\tau$. The total number of prisms in this decomposition of $LE(R)$ is linear in the complexity of $LE(R)$, which, by Lemma \ref{lemma: linear_envelope}, is $O(r)$. Each prism is defined by a constant number of pseudoplanes. Hence, Clarkson and Shor's analysis~\cite{CS89} can be applied to show that
	\[ \mathbb{E} \Big[ \sum_{i} |\Lambda_{\tau_i}|^{8/3}  \Big] = O \big( r \cdot (n/r)^{8/3} \big) = O(nk^{5/3}), \]
	where the sum is over all prisms $\tau_i$ in the above decomposition of $LE(R)$, where $\Lambda_{\tau_i}$ denotes the set of pseudoplanes of $\Lambda$ that intersect $\tau_i$, and where the expectation is over the random choice of $R$. We omit the easy and standard details of this application.	
	\end{proof}

	We remark that this bound is weaker than the bound $O(nk^{3/2})$ established in \cite{SST99}, which was obtained using a more refined lower bound argument than the one based on the Crossing Lemma. We use the weaker analysis because of the generalization of Euler's formula to arrangements of pseudolines, due to Tamaki and Tokuyama~\cite{TT03}, which allows us to extend our analysis to the case of pseudoplanes, as described in details in the previous subsection.
	
	\section{Discussion}
	\label{sec: discussion}
	In this paper we have shown that, for any set $\Lambda$ of $n$ surfaces in $\mathbb{R}^3$ that form a family of pseudoplanes, in the sense of satisfying properties (\romannum{1})--(\romannum{4}) of Section \ref{sec: pseudoplanes}, the complexity of the $k$-level of $\A(\Lambda)$ is $O(nk^{5/3})$. Our analysis is based on ingredients from the technique of \cite{SST99}, for the primal version of bounding the number of $k$-sets in a set of $n$ points in $\mathbb{R}^3$. The upper bound established in \cite{SST99} is $O(nk^{3/2})$, and is thus better than the bound we obtain here, for the case of pseudoplanes. The main reason for following this weaker analysis is the availability of the result of Tamaki and Tokuyama~\cite{TT03} on diamond-free graphs in arrangements of pseudolines, which leads to an extended dual version of the Crossing Lemma. It is definitely an intriguing, hopefully not too difficult, challenge to extend, to the case of pseudoplanes, a dual version of the sharper analysis in \cite{SST99}.
	
	Another line of research is to relax one or more of properties (\romannum{1})--(\romannum{4}), that define a family of pseudoplanes, as described in Section \ref{sec: pseudoplanes}, with the goal of extending our analysis and obtaining nontrivial bounds for the complexity of the $k$-level in arrangements of more general surfaces. Property (\romannum{4}) seems to be the most restrictive property among the four, namely requiring the $xy$-projections of all intersection curves from $\Lambda$ to form a family of pseudolines in the plane (although it trivially holds for planes). The main use of this property in our analysis is in proving a generalized dual version of the Lov\'asz Lemma (Lemma \ref{lemma: pseudo_lovasz}), as it (a) facilitates the applicability of topological sweeping, and (b) allows us to exploit the extended notion of antipodality, as in Lemma \ref{lemma: pseudo-antipodality}.
It is an interesting challenge to find refined techniques that can extend this analysis to situations where the arrangement within the curtain is not an arrangement of pseudolines. One open direction is to find a different proof technique of the Lov\'asz Lemma that is not based on sweeping. This would also be very interesting for the original case of planes (or of lines in the plane).
	
	In full generality, in studying the complexity of a level in an arrangement of more general surfaces, how far can we relax the constraints that these surfaces must satisfy in order to enable us to obtain sharp (significantly subcubic) bounds on the complexity of a level?
	
Finally, can our technique be extended to higher dimensions? For example, can we obtain a sharp bound for (suitably defined) pseudo-hyperplanes in four dimensions, similar to the bound in Sharir~\cite{S11} (or in~\cite{MSSW06}) for $k$-sets in $\mathbb{R}^4$?

\nocite{*}



\bibliography{references}

\end{document}